\documentclass{lmcs} 
\pdfoutput=1
\usepackage[utf8]{inputenc}

\usepackage{lastpage}
\lmcsdoi{21}{1}{11}
\lmcsheading{}{\pageref{LastPage}}{}{}%
{Feb.~12,~2024}{Jan.~30,~2025}{}


\usepackage{hyperref,amsmath,amssymb}
\usepackage{mathtools}
\usepackage{tikz,comment}
\usetikzlibrary{automata, arrows.meta, positioning,decorations.pathmorphing}

\newcommand{\Support}[1]{\mathrm{Support}_{\mathrm{B}}(#1)}
\newcommand{\SupportC}[1]{\mathrm{Support}_{\mathrm{C}}(#1)}

\newcommand{\tempB}{\mathbf{B}}
\newcommand{\tempC}{\mathbf{C}}
\newcommand{\tempT}{\mathbf{T}}

\newcommand{\shape}{\mathcal{S}}

\newcommand{\cS}{\mathcal{S}}

\newcommand{\satisfies}{\;\raisebox{0.25ex}{$\scriptstyle \models$}\;}

\newcommand{\A}{\mathcal{A}}
\newcommand{\B}{\mathcal{B}}
\newcommand{\C}{\mathcal{C}}
\renewcommand{\P}{\mathcal{P}}
\newcommand{\G}{\mathcal{G}}
\newcommand{\D}{\mathcal{D}}

\newcommand{\T}{\mathcal{T}}

\newcommand{\intinter}[2]{[\![#1,#2]\!]}
\newcommand{\proba}{\mathbb{P}}
\newcommand{\uprod}{\odot}

\newcommand{\dback}{d^{-}}
\newcommand{\backT}{\mathfrak{B}}
\newcommand{\treeT}{\mathfrak{T}}

\begin{document}

\title[{Random Deterministic Automata With One Added Transition}]{{Random Deterministic Automata}\texorpdfstring{\\}{} With One Added Transition}
\titlecomment{This article is the extended and revised version of the article~\cite{STACS23} published in the proceedings of the conference STACS'23 }
\thanks{}	

\author[A.~Carayol]{Arnaud Carayol\lmcsorcid{0009-0008-2763-2821}}[a]
\author[P.~Duchon]{Philippe Duchon\lmcsorcid{0009-0009-4735-0781}}[b]
\author[F.~Koechlin]{Florent Koechlin\lmcsorcid{0000-0002-5576-4847}}[c]
\author[C.~Nicaud]{Cyril Nicaud\lmcsorcid{0000-0002-8770-0119}}[a]

\address{Univ Gustave Eiffel, CNRS UMR 8049, LIGM, F-77454 Marne-la-Vallée, France}	
\email{arnaud.carayol@univ-eiffel.fr, cyril.nicaud@univ-eiffel.fr}  

\address{Univ. Bordeaux,  CNRS UMR 5800, LaBRI, F-33400 Talence, France}	
\email{duchon@labri.fr}  

\address{Univ. Sorbonne Paris Nord, LIPN, CNRS UMR 7030, F-93430 Villetaneuse, France}	
\email{koechlin@lipn.fr}




\begin{abstract}
Every language recognized by a non-deterministic finite automaton can be recognized by a deterministic automaton, at the cost of a potential increase of the number of states, which in the worst case can go from $n$ states to $2^n$ states. 
In this article, we investigate this classical result in a probabilistic setting where we take a deterministic automaton with $n$ states uniformly at random and add just one random transition. These automata are almost deterministic in the sense that only one state has a non-deterministic choice when reading an input letter. In our model, each state has a fixed probability to be final.
We prove that for any $d\geq 1$, with non-negligible probability the minimal (deterministic) automaton of the language recognized by such an automaton has more than $n^d$ states; as a byproduct, the expected size of its minimal automaton grows faster than any polynomial. Our result also holds when each state is final with some probability that depends on $n$, as long as it is not too close to $0$ and $1$, at distance at least $\Omega(\frac1{\sqrt{n}})$ to be precise, therefore allowing  models with a sublinear number of final states in expectation.
\end{abstract}

\maketitle

\section{Introduction}\label{S:one}

A fundamental result in automata theory is that deterministic complete finite state automata recognize the same languages as non-deterministic finite state automata. This result can be established using the classical (accessible) subset construction~{\cite{RabinS59,MF71,Hopcroft79}}: starting with 
a non-deterministic automaton with $n$ states, one can build a deterministic automaton with at most $2^n$ states that recognizes the same language. This upper bound is tight; there are regular languages recognized by an $n$-state non-deterministic automaton whose minimal automaton, i.e. the smallest deterministic and complete automaton that recognizes the language, has $2^n$ states. The number of states of the minimal automaton of a regular language is called its \emph{state complexity}.
Figure~\ref{fig:intro examples} shows two $n$-state non-deterministic automata with somewhat similar shape, {but} whose languages {${\mathcal{L}}_\ell$ and ${\mathcal{L}}_r$} have very different state complexities. In both automata, there is only one non-deterministic choice, at the initial state.

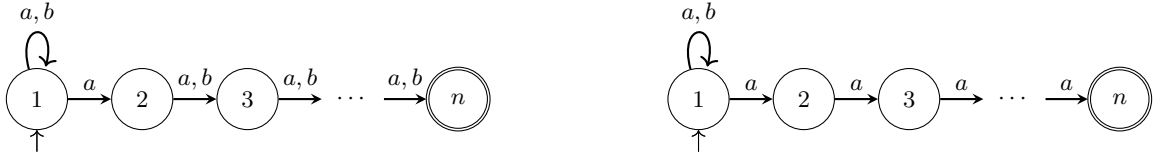
\begin{figure}[h]
{\footnotesize
    \begin{tikzpicture} [node distance = 1.4cm, on grid, scale=.7]
        \node (q1) [state,initial, initial below, initial text=] {$1$};
        \node (q2) [state, right of = q1] {$2$};
        \node (q3) [state, right of = q2] {$3$};
        \node (dots) [state,draw=none, right of = q3] {$\cdots$};
        \node (qn) [state, accepting, right of = dots] {$n$};

        \draw [-stealth, thick] (q1) edge [loop above] node[above] {$a,b$} (q1);
        \draw [-stealth, thick] (q1) edge node [above] {$a$} (q2);
        \draw [-stealth, thick] (q2) edge node [above] {$a,b$} (q3);
        \draw [-stealth, thick] (q3) edge node [above] {$a,b$} (dots);
        \draw [-stealth, thick] (dots) edge node [above] {$a,b$} (qn);
    \end{tikzpicture}
    \hfill
    \begin{tikzpicture} [node distance = 1.4cm, on grid, scale=.7]
        \node (q1) [state,initial, initial below, initial text=] {$1$};
        \node (q2) [state, right of = q1] {$2$};
        \node (q3) [state, right of = q2] {$3$};
        \node (dots) [state,draw=none, right of = q3] {$\cdots$};
        \node (qn) [state, accepting, right of = dots] {$n$};

        \draw [-stealth, thick] (q1) edge [loop above] node[above] {$a,b$} (q1);
        \draw [-stealth, thick] (q1) edge node [above] {$a$} (q2);
        \draw [-stealth, thick] (q2) edge node [above] {$a$} (q3);
        \draw [-stealth, thick] (q3) edge node [above] {$a$} (dots);
        \draw [-stealth, thick] (dots) edge node [above] {$a$} (qn);
    \end{tikzpicture}
    \caption{ \label{fig:intro examples} On the left, a non-deterministic automaton with $n$ states recognizing the language ${\mathcal{L}}_\ell = \Sigma^*a\Sigma^{n-2}$. On the right, a non-deterministic automaton with $n$ states recognizing the  language ${\mathcal{L}}_r = \Sigma^*a^{n-1}$. The minimal automaton of ${\mathcal{L}}_\ell$ has $2^{n-1}$ states, whereas the one of ${\mathcal{L}}_r$ has $n$ states.
    }
}
\end{figure}

In this article, we address the  following question: if we take a random $n$-state deterministic automaton and add just one random transition, what can be said about the state complexity of the resulting recognized language? Does it hugely increase as for $\mathcal{L}_\ell$, or does it remain small as for $\mathcal{L}_r$?

From~\cite{sportiello12}, we know that with high probability, the state complexity of the language recognized by a $n$-state deterministic automaton taken uniformly at random is linear. This is important as it implies that the corresponding distribution on regular languages is not degenerated: this contrasts with the case of random regular expressions where the expected state complexity of the described regular languages is constant~\cite{expressions21}, which means that the induced distribution on regular languages is concentrated on a finite number of languages.

To be more precise, our formal setting in this article is the following.  Let $\Sigma=\{a,b,\ldots\}$ be a finite alphabet with $k\geq 2$ letters. For any $n\geq 1$, we consider the uniform distribution on deterministic and complete automata on $\Sigma$, with $\{1,\ldots, n\}$ as their set of states; the initial state is picked uniformly at random, and the action of the letters on the set of states are $k$ uniform and independent random mappings. We also pick uniformly at random and independently two states $p$ and $q$, and add a transition $p\xrightarrow{a}q$, if it is not already there. Finally each state is final with a given fixed probability  $f\in(0,1)$, independently. 
Hence in this model, the expected number of final states of an almost deterministic automaton is $f\times n$. Our results still hold if we allow the probability $f$ of being final to depend on the size $n$ of the automaton, provided that  $f_n$ has a distance to 0 and 1 in $\Omega(\frac{1}{\sqrt{n}})$. This allows us to consider probabilistic models in which random automata have an expected number of final states that is as low as $\Theta(\sqrt{n})$.

Our main result is that for any $d\geq 1$ there exists a constant $c_d>0$ such that the state complexity of the language of such a random almost deterministic automaton is greater than $n^d$ with probability at least $c_d$, for $n$ sufficiently large. That is, for any polynomial $P$, there is a non-negligible probability that the state complexity of the language of a random automaton is greater than $P(n)$: we will say that the state complexity is \emph{super-polynomial} with \emph{visible probability}. As a direct consequence, the expected state complexity is super-polynomial.

It should be noted that with the same random models for deterministic automata, one cannot hope to replace visible probability in our results with a probability that converges to~$1$ (i.e., with high probability). Indeed random automata have, with high probability, a constant fraction of states that are not accessible from the initial state \cite{grusho73}; if the source of the added transition is not accessible from the initial state, the added transition does not impact the recognized language, whose state complexity is therefore at most equal to $n$. Thus, we make no effort in the present article to optimize our probabilistic lower bounds. See the conclusion for a more advanced discussion on this topic.

\noindent\textbf{Related work.} {A subclass of 
almost deterministic automata has already been considered in~\cite{BjorklundM08} to study the worst-case complexity class for the non-deterministic minimization problem for NFAs.}

The study of random deterministic automata can be traced back to the work of Grusho on the size of the accessible part~\cite{grusho73}: he established that, with high probability, a constant proportion of the states are accessible from the initial state. 
He also shows that with high probability there is a unique terminal strongly connected component\footnote{{A strongly connected component (SCC) is a set of vertices, maximal for inclusion, in which every vertex is reachable for each other vertex. A SCC is terminal if no edge leave the SCC.}} (SCC) of size approximately $\nu_k n$, for some $\nu_k>\frac12$ that only depends on the size $k$ of the alphabet. More structural results on the underlying graph of a random deterministic automaton were established in the work of Carayol and Nicaud~\cite{CarayolN12}, with a local limit law for the size of the accessible part and an application to random generation of accessible determistic automata. 
More recently, {Cai and Devroye gave in \cite{devroye17},} a fine-grained analysis of what is happening outside the large strongly connected component. In~\cite{addario2020},  Addario-Berry,  Balle and  Perarnau gave a precise analysis of the diameter of a random deterministic automaton, showing in particular that it is logarithmic. 
We will use some of these results in this article, for instance that there is a unique largest terminal strongly connected component with high probability. To deal with the restriction to states accessible from the initial state in the powerset construction, we also use  a result of~\cite{devroye17} which yields that, with high probability, there are no cycles of length $\Omega(\log n)$ outside this terminal strongly connected component.

All these results on random automata focus on the underlying graph of the transition structures, without saying much about the recognized languages, or on the average complexity of textbook algorithms on automata. Some results were established in this direction: the probability that a random accessible automaton is minimal was studied by Bassino, David and Sportiello~\cite{sportiello12}, the analysis of minimization algorithms by Bassino, David and Nicaud~\cite{Moore12,David12}, etc.  

In another direction, more recently, several articles studied the synchronization of random automata~\cite{berlinkov16,synchro19}, until the very recent work of Chapuy and Perarnau~{\cite{ChapuyP23}}, establishing that most deterministic automata are synchronizing, with a word of length $O(\sqrt{n}\log n)$ and Martinsson~\cite{martinsson23} established that for any $\epsilon>0$, there exists a synchronizing word of length $O(\epsilon^{-1}\sqrt{n}\log n)$ with probability at least $1-\epsilon$. 
We refer the interested reader to the survey of Nicaud~\cite{survey14} for an overview  on random deterministic automata.

To our knowledge, there is no well-established random model for non-deterministic automata. It is not an easy task to obtain a satisfactory model: for instance, the uniform distribution is degenerated and produces languages with state complexity one or two with high probability (see the discussion in Section~\ref{sec:dense models}). 
Applying the powerset construction to the mirror of a random deterministic automaton was studied by
De Felice and Nicaud~\cite{sven13,sven16}, in order to analyze the average case  complexity of  Brzozowski's state minimization algorithm. As in the present article, they studied the determinization procedure of random automata, but for a model that is very different from ours: they consider the mirror of a uniform random deterministic automaton, obtained by reversing the transitions and swapping the initial and final states. In particular, with high probability, there is a linear number of states having a non-deterministic choice in their setting. Another natural model would be to use a critical Erd\H{o}s-Rényi~\cite{ER60} digraph for each letter, which would also result in a linear number of states having a non-deterministic choice. In this article, we choose a random model with the minimum amount of non-determinism by adding just one transition to a uniform deterministic automaton, and establish that we likely have a combinatorial explosion already in this case.

This article is the full version of the extended abstract published in the proceedings of the STACS conference~\cite{STACS23}. It contains all the omitted proofs. We also introduce the notion of templates in Section~\ref{sec:template} to simply handle computation under natural conditional properties on automata, and reshape completely the most technical part, Section~\ref{sec:backward}, to exhibit an approximation by classical Galton-Watson processes. This last approach is interesting on its own for future works on random automata.

\section{Definitions and notations}\label{sec:definitions}

The cardinality of a finite set $E$ is denoted by $|E|$.
For any $n \geq 1$, let $[n]=\{1,\ldots,n\}$. If $x,y\in\mathbb{R}$ with $x\leq y$, let $\intinter{x}{y}=[x,y]\cap\mathbb{Z}$ be the set of integers that are between $x$ and $y$.
Let $\mathcal{E}$ be a set equipped with a size function $s$ from $\mathcal{E}$ to $\mathbb{Z}_{\geq 0}$, and let ${\mathcal{E}}_{n}$ denote the elements of $\mathcal{E}$ having size $n$. A property $X$ on $\mathcal{E}$ (that is, a subset of $\mathcal{E}$ viewed as the set of elements for which the property holds) holds with \emph{visible probability} if there exists some constant $c>0$ such that, for $n$ sufficiently large, ${\mathcal{E}}_n$ is non-empty and $\mathbb{P}(X)\geq c$ for the uniform distribution on ${\mathcal{E}}_n$. By a slight abuse of notation, if $X$ is a random variable ${\mathcal{E}}\rightarrow \mathbb{Z}_{\geq 0}$ we say that for the uniform distribution on $\mathcal{E}$, $X$ is \emph{super-polynomial with visible probability} when for any $d\geq 1$, there exists a constant $c_d>0$, such that for $n$ sufficiently large, ${\mathcal{E}}_n\neq\emptyset$ and $\mathbb{P}(X\geq n^d)\geq c_d$.

For any real {number} $\lambda>0$, we denote by $\mathrm{Poi}(\lambda)$ the Poisson random variable of parameter $\lambda$, whose support is $\mathbb{Z}_{\geq 0}$ and defined by 
$\proba(\mathrm{Poi}(\lambda)=k)={\frac{\lambda^k}{k!}e^{-\lambda}}$, for all $k\in\mathbb{Z}_{\geq 0}$.

Recall that if $u$ and $v$ are two words on an ordered alphabet $\Sigma$, $u$ is \emph{smaller than $v$ for the length-lexicographic order} if $|u|<|v|$ or they have same length and $u<_{\text{lex}} v$ for the lexicographic order.

Throughout the article, the set of states of an automaton with $n$ states  will always be $[n]$, with the exception of the powerset construction recalled just below. The alphabet will always be $\Sigma=\{a,b\}$, except in the statement of our main theorem, where we allow larger alphabets as it is trivially generalized to this case. Hence, in our setting, a \emph{deterministic (and complete) automaton} is just a tuple $(n,\delta,F)$, where $F\subseteq[n]$ is the \emph{set of final states} and $\delta$ is the \emph{transition function}, a mapping from $[n]\times\Sigma$ to $[n]$. 
We will often write  $s\xrightarrow{\alpha} t$ instead of $\delta(s,\alpha)=t$, for $s,t\in[n]$ and $\alpha\in\Sigma$, and call this an $\alpha$-transition or a transition. The transition function is classically extended to sets of states by setting $\delta(X,\alpha) = \{\delta(s,\alpha): s\in X\}$, for $X\subseteq [n]$, and to words by setting inductively $\delta(s,w) = s$ if $w$ is the empty word $\varepsilon$ and $\delta(s,w\alpha) = \delta(\delta(s,w),\alpha)$. We will not need to specify the \emph{initial state} {until the end of the proof of Theorem~\ref{th:main}}; when we finally do, it will be generated uniformly at random and independently in $[n]$. 
Final states are only used in the last part of our proof, so to ease the presentation, we define a \emph{deterministic (and complete) transition structure} as being an automaton with neither initial nor final states: it is given by a pair $(n,\delta)$ where $n$ is the number of states and $\delta$ is the transition function. 

An \emph{almost deterministic automaton} $(n,\delta,F,p\xrightarrow{a} q)$ is a deterministic automaton $(n,\delta,F)$ in which we add the additional $a$-transition $p\xrightarrow{a} q$. Similarly, an \emph{almost deterministic transition structure} $(n,\delta,p\xrightarrow{a} q)$ is a deterministic transition structure $(n,\delta)$ in which we add the additional $a$-transition $p\xrightarrow{a} q$. For any $\alpha\in\Sigma$ and any $r\in[n]$, the transition function $\gamma$ of an almost deterministic automaton $(n,\delta,F,{p\xrightarrow{a} q})$ (or almost deterministic transition structure)  is therefore defined by $\gamma(r,\alpha)=\{\delta(r,\alpha)\}$ if $(r,\alpha)\neq (p,a)$ and 
$\gamma(p,a)=\{\delta(p,a), q\}$.
These automata or transition structures can be deterministic, in the case where $\delta(p,a)=q$.

The  \emph{powerset automaton} $\B$ of an almost deterministic automaton $\A=(n,\delta,F,p\xrightarrow{a} q)$, with a transition function $\gamma$, is a deterministic automaton $\B$ with states in $2^{[n]}$ and transition function $\gamma$ extended to sets, as defined above. If we add an initial state $i_0$ to $\A$, the initial state of $\B$ is $\{i_0\}$ and {$\B$} recognizes the same language as $\A$ when a state $X$ of $\B$ is final if and only if at least one of its element is final in $\A$. 
This construction can be restricted to the states of $\B$ that are accessible from the initial state $\{i_0\}$ while still recognizing the same language; we call this automaton the  \emph{accessible powerset automaton} of~$\A$. 

Recall that two states $r$ and $s$ in a deterministic automaton $\A$ are \emph{equivalent} if the languages recognized by moving the initial state to $r$ or to $s$ are equal. The \emph{minimal automaton} of a regular language $\mathcal{L}$ is the deterministic complete automaton with 
the smallest number of states that recognizes $\mathcal{L}$. The number of states of the minimal automaton of $\mathcal{L}$ is called the \emph{state complexity} of $\mathcal{L}$. We will use the following classical property, {which follows from the fact that the minimal DFA is obtained by merging equivalent states in any DFA accepting the language~\cite[Theorem~3.11]{Hopcroft79}}:
 \begin{prop}
 If there is a set of accessible states $X$ in a deterministic automaton $\A$ such that the states of $X$ are pairwise non-equivalent, then $\A$ has state complexity at least $|X|$. 
 \end{prop}

\section{Main statement and proof outline}\label{sec:main}

Our main result is that the state complexity of the language recognized by a random almost deterministic automaton is super-polynomial with visible probability, when for each $n$, each state is final, independently, with some probability $f_n$ that is not too close to either $0$ or $1$, as precised in the statement:
\begin{thm}\label{th:main}
Let $\Sigma$ be an alphabet with at least two letters. Let $f_n$ be a map from $\mathbb{Z}_{\geq 1}$ to $(0,1)$ such that there exists a constant $\alpha>0$ such that $f_n \geq \frac{\alpha}{\sqrt{n}}$ and $1-f_n \geq  \frac{\alpha}{\sqrt{n}}$ {for $n$ sufficiently large}.
Consider an almost deterministic $n$-state transition structure $\A$ on $\Sigma$ taken uniformly at random. Each state of $\A$ is then taken to be final with probability $f_n$, independently of everything else. Then, with visible probability, the language recognized by $\A$ has super-polynomial state complexity.
\end{thm}

First, we observe that if $\Gamma\subseteq \Sigma$ are two non-empty alphabets and
 if $\mathcal{L}$ is a regular language on $\Sigma$, then the state complexity of $\mathcal{L}$ is at least the state complexity of ${\mathcal{L}}\cap \Gamma^*$. As a consequence, it is sufficient to establish Theorem~\ref{th:main} for a two-letter alphabet, and from now on, we fix $\Sigma=\{a,b\}$.

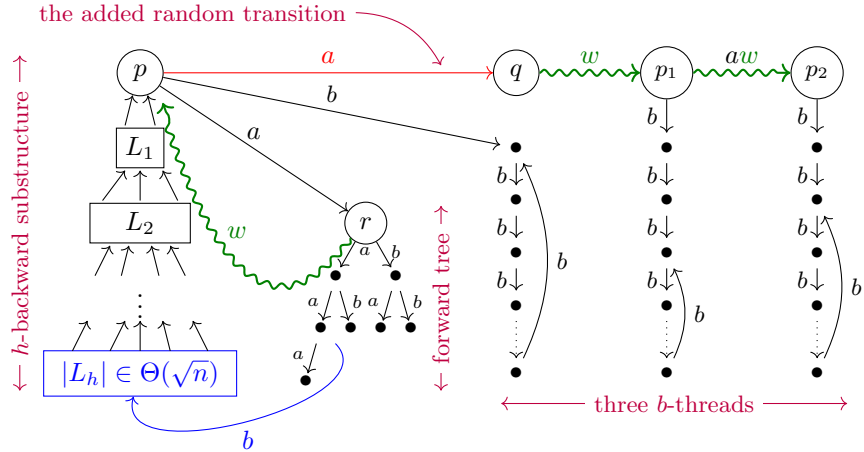
\begin{figure}[t]
    \centering
    {\small
\begin{tikzpicture}
	\node[draw, circle] (s) at (0,0) {$p$};
	
	\node[draw] (b1) at (0,-1) {$L_1$};
	\node[draw] (b2) at (0,-2) {$\quad L_2\quad$};
	\node (dots1) at (0,-3) {$\vdots$};
	\node[draw,blue] (btau) at (0,-4) {$\ |L_h|\in\Theta(\sqrt{n})\ $};

	\draw[->] (.2,-0.7) -- (.1,-0.3);
	\draw[->] (-.2,-0.7) -- (-.1,-0.3);
	
	\draw[->] (0,-1.7) -- (0,-1.3);
	\draw[->] (-.5,-1.7) -- (-.2,-1.3);
	\draw[->] (.5,-1.7) -- (.3,-1.3);

	\draw[->] (-.25,-2.7) -- (-.15,-2.3);
	\draw[->] (.25,-2.7) -- (.15,-2.3);
	\draw[->] (-.6,-2.7) -- (-.4,-2.3);
	\draw[->] (.6,-2.7) -- (.4,-2.3);
	
	\draw[->] (0,-3.7) -- (0,-3.3);
	\draw[->] (-.4,-3.7) -- (-.2,-3.3);
	\draw[->] (.4,-3.7) -- (.3,-3.3);
	\draw[->] (-.9,-3.7) -- (-.7,-3.3);
	\draw[->] (.9,-3.7) -- (.7,-3.3);

    	\draw[thin,<->,purple] (-1.6,-4.23) -- node[rotate=90,fill=white]{\footnotesize $h$-backward substructure} (-1.6,0.23);

	\node[draw,circle] (r) at (3,-2) {$r$};
	\draw[->] (s) -- node[above]{$a$} (r);
	\node(ra) at (2.6, -2.7) {$\bullet$}; \draw[->,shorten >=-3pt] (r) -- node[right]{\tiny $a$} (ra);
	\node (rb) at (3.4, -2.7) {$\bullet$}; \draw[->,shorten >=-3pt] (r) -- node[right]{\tiny $b$} (rb);
	\node (raa) at (2.4, -3.4) {$\bullet$}; \draw[->,shorten >=-3pt] (ra) -- node[left]{\tiny $a$} (raa);
	\node (rab) at (2.8, -3.4) {$\bullet$}; \draw[->,shorten >=-3pt] (ra) -- node[right]{\tiny $b$} (rab);
	\node (rba) at (3.2, -3.4) {$\bullet$}; \draw[->,shorten >=-3pt] (rb) -- node[left]{\tiny $a$} (rba);
	\node (rbb) at (3.6, -3.4) {$\bullet$}; \draw[->,shorten >=-3pt] (rb) -- node[right]{\tiny $b$} (rbb);
	\node (raaa) at (2.2, -4.1) {$\bullet$}; \draw[->,shorten >=-3pt] (raa) -- node[left]{\tiny $a$} (raaa);
	\draw[blue,->]  (raa) edge[bend left=120] node[below]{$b$} (btau);

	\draw[thin,<->,purple] (4,-4.23) -- node[rotate=90,fill=white]{\footnotesize forward tree} (4,-1.77);

	\node[draw,circle] (t) at (5,0) {$q$};
	\draw[->,red] (s) -- node[above]{$a$} (t);

	\node[draw,circle] (s1) at (7,0) {\footnotesize $p_1$};
	\draw[->,thick,green!50!black,decorate,decoration={snake,amplitude=.4mm,segment length=2mm,post length=1mm}] (t) -- node[above]{$w$} (s1);

	\node[draw,circle] (s2) at (9,0) {\footnotesize $p_2$};
	\draw[->,thick, green!50!black,decorate,decoration={snake,amplitude=.4mm,segment length=2mm,post length=1mm}] (s1) -- node[above]{${\color{black}a}w$} (s2);

	\draw[->,green!50!black, thick,decorate,decoration={snake,amplitude=.4mm,segment length=2mm,post length=1mm}] (2.8,-2.2) .. controls (1.3,-4.2) and  (.4,-.6) .. node[right]{$w$} (.3,-.4);
	
	\node (p01) at (5,-1) {$\bullet$}; \draw[->] (s) -- node[above]{\footnotesize $b$} (p01);
	\node (p02) at (5,-1.7) {$\bullet$}; \draw[->] (p01) -- node[left]{\footnotesize $b$} (p02);
	\node (p03) at (5,-2.4) {$\bullet$}; \draw[->] (p02) -- node[left]{\footnotesize $b$} (p03);
	\node (p04) at (5,-3.1) {$\bullet$}; \draw[->] (p03) -- node[left]{\footnotesize $b$} (p04);
	\node (p0n) at (5,-4) {$\bullet$}; \draw[->,dotted] (p04) -- (p0n);
	\draw[->] (p0n) edge[bend right=25] node[right]{\footnotesize $b$} (p01);

	\node (p11) at (7,-1) {$\bullet$}; \draw[->] (s1) -- node[left]{\footnotesize $b$} (p11);
	\node (p12) at (7,-1.7) {$\bullet$}; \draw[->] (p11) -- node[left]{\footnotesize $b$} (p12);
	\node (p13) at (7,-2.4) {$\bullet$}; \draw[->] (p12) -- node[left]{\footnotesize $b$} (p13);
	\node (p14) at (7,-3.1) {$\bullet$}; \draw[->] (p13) -- node[left]{\footnotesize $b$} (p14);
	\node (p1n) at (7,-4) {$\bullet$}; \draw[->,dotted] (p14) -- (p1n);
	\draw[->] (p1n) edge[bend right=25] node[right]{\footnotesize $b$} (p13);

	\node (p21) at (9,-1) {$\bullet$}; \draw[->] (s2) -- node[left]{\footnotesize $b$} (p21);
	\node (p22) at (9,-1.7) {$\bullet$}; \draw[->] (p21) -- node[left]{\footnotesize $b$} (p22);
	\node (p23) at (9,-2.4) {$\bullet$}; \draw[->] (p22) -- node[left]{\footnotesize $b$} (p23);
	\node (p24) at (9,-3.1) {$\bullet$}; \draw[->] (p23) -- node[left]{\footnotesize $b$} (p24);
	\node (p2n) at (9,-4) {$\bullet$}; \draw[->,dotted] (p24) -- (p2n);
	\draw[->] (p2n) edge[bend right=25] node[right]{\footnotesize $b$} (p22);

	\draw[thin,<->,purple] (4.8,-4.4) -- node[fill=white]{\footnotesize three $b$-threads} (9.4,-4.4);

    \node[purple] (text) at (.7,.8) {\footnotesize the added random transition};
    \draw[purple, thin] (text.east) edge[->,bend left] (4,.1);

	\end{tikzpicture}
}
    \caption{Illustration of the proof sketch of Theorem~\ref{th:main}. On the left, the $h$-backward substructure from~$p$ that is detailed in Section~\ref{sec:backward}. It has size $\Theta(\sqrt{n})$ and contains  $\Theta(\sqrt{n})$ extremal leaves (i.e. leaves in its last level $h$) to be valid. In the center, the forward tree from $r$, described in Section~\ref{sec:forward}; it is a breadth-first traversal that is valid if it  hits an extremal leaf of the backward substructure before $\sqrt{n}$ states are examined. On the right the $b$-threads introduced in Section~\ref{sec:b-cycles}, obtained by reading $b$'s from the $p_i$'s; they are valid if they are made of previously unseen states and are pairwise disjoint until they cycle back on themselves, forming a $b$-cycle of length in $\intinter{\sqrt{n}}{2\sqrt{n}\,}$. \label{fig:proof sketch}}
\end{figure}

The proof of Theorem~\ref{th:main} consists in identifying a structure and several constraints (see Figure \ref{fig:proof sketch}) that guarantee that when performing the accessible powerset construction and adding a random set of final states, we have sufficiently many pairwise non-equivalent states. At each step, we add a new constraint on top of those we already have, and we have to ensure that these constraints are still satisfied by sufficiently many almost deterministic transition structures. 
A convenient way to sketch the proof is to consider that we start with $n$ states and no transition, and add random transitions when needed, on the fly. More precisely, our proofs can be seen as the description of an algorithm that tries to expose the required structure by performing two types of queries on the set of still unknown transitions: either we ask what the destination of a given transition is, or we ask for all the transitions that have a given state as their destination. Thus, at any point in the algorithm, conditioned on the results of all previous queries, the destinations of all still {unexplored} transitions are independent and uniform among the set of states for which we have not performed the second type of query. These  constraints and associated conditional probabilities are formalized in a unified way in Section~\ref{sec:template}, introducing and using the notion of \emph{templates}.

We use this to handle the computations and establish that our algorithm has a non-negligible probability of success. We also have two random states $p$ and $q$ and the transition $p\xrightarrow{a}q$ will be added at some point. 
 {Let $d \geq 1$. We describe} the  main steps of the proof below. In this high level description, recall that $\delta$ refers to the transition function of the deterministic base of the almost deterministic automaton being generated. 

\begin{enumerate}
    \item Generate $r=\delta(p,a)$, the target of the $a$-transition starting from $p$ in the deterministic transition structure. With visible probability, $r\neq q$ and there is a word $w$ of length $\Theta(\log n)$ such that $\delta(r,w)=p$, which can be found by generating $O(\sqrt{n})$ random transitions. We also assume that the $b$-transition starting at $p$ is still unset. This step is the most technical, we explore backward from $p$ and forward from $r$ until we reach a common state.
    \item Assuming such a $w$ is found, 
    we  iteratively generate the transitions starting from $q$ and following the word
    $w(aw)^{d-1}$, and  {ask for} the target of each such transition to be a state that was not previously seen in the whole process. This happens with visible probability.
    \item Let $p_0=p$ and $p_i = \delta(q,w(aw)^{i-1})$ for $i\in [d]$. At this point, if we add the transition $p\xrightarrow{a}q$, and denote $\hat\delta$ the resulting non-deterministic transition function, we have $\hat\delta(\{p\},(aw)^d) = \{p_0,p_1,\ldots,p_d\}$, and the outgoing $b$-transitions of the $p_i$'s are still unset. Then, for each $p_i$, we iteratively generate the $b$-transitions $\delta(p_i,b)$, $\delta(p_i,bb)$, \ldots until we cycle after $\lambda_i$ steps. This process is considered successful if we do not use an already set $b$-transition and if  the $d+1$ cycles are pairwise disjoint. We furthermore {require that} the $\lambda_i$'s are all in $\intinter{2\sqrt{n}}{3\sqrt{n}\,}$. All these properties  happen with visible probability.
    \item 
   The set $\{p_0,\ldots, p_{d}\}$ is now composed of $d+1$ different states, and reading $b$'s from each $p_i$ eventually ends in a $b$-cycle of length $\ell_i$.
    Given the $\lambda_i$'s, each $\ell_i$ is a uniform element of $[\lambda_i]$, and they are independent.  Our precise requirements ensure that once met, the $\ell_i$'s are uniform and independent elements of $\intinter{\sqrt{n}}{2\sqrt{n}}$. We now {require that} the $\ell_i$'s are pairwise coprime, which also happens with visible probability.
    \item If everything worked so far, we can guarantee that $\{p\}$ is accessible in the subset construction, still with visible probability: we use the fact that with high probability, all cycles with length in $\Omega(\log n)$ are accessible in a random deterministic automaton~\cite{devroye17}. By construction the cycle around $p$ labeled $aw$ built at Step~(1) has length $\Theta(\log n)$, hence $p$ is accessible with high probability.
    \item Hence, at this stage, in the accessible powerset construction applied to the almost deterministic transition structure there is a $b$-cycle of length $\prod_{i=0}^d\ell_i = \Omega(n^{\frac{d+1}2})$. We now randomly determine which states are final. If we consider a $b$-cycle alone in the automaton, of length $\Omega(\sqrt{n})$, its states are pairwise non-equivalent with visible probability as soon as the probability $f_n$ that a state is final is not too close to either $0$ or $1$, which we assumed in our model. This property happens to be preserved when building the product automaton for the union of  one-letter cycles, provided their lengths are pairwise coprime. Consequently, the large $b$-cycles in the powerset construction are made of pairwise non-equivalent states with visible probability.

\end{enumerate}
The first steps of the proof sketch are depicted in Figure~\ref{fig:proof sketch}. More details and notations will be introduced in the next section. 

\bigskip
We also state the following consequence for the expected state complexity of a random almost deterministic automaton. Note this corollary is weaker than the result of Theorem~\ref{th:main}.
Indeed, as its state complexity could be up to $2^n$, a 
 negligible proportion of automata could still contribute sufficiently to have a super-polynomial expected state complexity. 
\begin{cor}\label{cor:main}
Under the conditions of Theorem~\ref{th:main}, the expected state complexity of the language recognized by $\A$ grows faster than any polynomial in $n$.
\end{cor}

\begin{proof}
Let $S$ be the random variable that maps a random automaton to the state complexity of the language it recognizes. By Theorem~\ref{th:main}, for $n$ sufficiently large, $\mathbb{P}(S\geq n^d) \geq c_d\, n^d$ for some $c_d>0$. Hence, $\mathbb{E} [S] \geq n^d\ \mathbb{P}(S\geq n^d) \geq c_d\, n^d$, concluding the proof.
\end{proof}

\section{Templates}\label{sec:template}
In this section we introduce the notion of templates, which will be used to describe families of transition structures and will allow to handle the  computations of conditional probabilities in a unified way.

A \emph{template} $\T$ is a pair $(n,\lambda)$, where $\lambda$ is a function from $[n]\times\Sigma$ to $2^{[n]}$, the subsets of $[n]$, such that for every $(x,\alpha)\in[n]\times\Sigma$, $\lambda(x,\alpha)$ is non-empty. 

From their definition, templates could be identified as non-deterministic and complete transition structures, but we only use them to encode restrictions. For $\alpha \in \Sigma$ and $x \in [n]$, $\lambda(x,\alpha)$ is seen as the set of allowed targets for the $\alpha$-transition outgoing from $x$.

We say that a transition structure $\A=(n,\delta)$ \emph{satisfies} a template $\T=(n,\lambda)$, denoted by $\A{\satisfies}\T$, if for all $(x,\alpha)\in[n]\times\Sigma$, $\delta(x,\alpha)\in\lambda(x,\alpha)$.
In particular, if $|\lambda(x,\alpha)|=1$ for each $(x,\alpha)\in[n]\times\Sigma$, then $\T=(n,\lambda)$ can be identified with the unique transition structure that satisfies $\T$.

The next lemma is used throughout the article to work with the uniform distribution on transition structures conditioned on satisfying a  template.
It is obtained by direct counting.

\begin{lem}\label{lm:proba template}
Conditioned on satisfying the template $\T=(n,\lambda)$, the uniform distribution on $n$-state transition structures $(n,\delta)$ is distributed the same way as generating the image of each $\delta(x,\alpha)$ uniformly at random and independently in $\lambda(x,\alpha)$. 
\end{lem}

In the sequel we define several sets of templates to capture the required properties at every step, following the informal description given in Section~\ref{sec:main}. Without going into the details at the moment, which will be given as needed, we will define:
\begin{itemize}
    \item $\tempB_n(p)$ as the templates that force the required conditions on the backward exploration from state $p$, defined in Section~\ref{sec:backward};
    \item $\tempC_n(p)$ as the templates built from $\tempB_n(p)$ that furthermore force  the existence of a small cycle around $p$, defined in Section~\ref{sec:forward};
    \item $\tempT_n(p,q)$ as the templates built from $\tempC_n(p)$ that furthermore force the existence of the path starting at $q$ and of the $b$-threads, as depicted in Figure~\ref{fig:proof sketch} and defined in Section~\ref{sec:b-cycles}.
\end{itemize}

\section{Backward substructure of a random transition structure}\label{sec:backward}

This section studies the shape of the backward substructure at logarithmic depth from a state $p \in [n]$ in a random transition structure of size $n$. The backward substructure at depth $h$ from a vertex $p$, formally defined in Subsection~\ref{sec:backward substructure}, is obtained by restricting the transition structure to the vertices that can reach $p$ in at most $h$ steps.

More precisely, we study the backward substructure at depth $h:=\lceil\log_2  \sqrt{n}\,\rceil$ and we want to show that with visible probability 
in a random transition structure, this backward substructure contains at most $d_1\sqrt{n}$ vertices with at least $d_2\sqrt{n}$ vertices at depth exactly $h$, for some positive constants $d_1$ and $d_2$. This result is formally stated in Proposition~\ref{pro:backward} and its proof is the aim of this section.

For a fixed maximal depth $m$, Cai and Devroye~\cite{devroye17} studied the stochastic process giving the number of vertices at depth at most $m$ in a random transition structure of size $n$ by providing a perfect coupling with a Galton-Watson process whose offspring distribution is $\mathrm{Bin}(2n,1/n)$.  As we work at logarithmic depth $\lceil\log_2  \sqrt{n}\,\rceil$, we are not able to provide such a tight description. 
However, inspired by their approach, we introduce in Subsection~\ref{sec:backward tree} a stochastic process $\G_h$ that builds what we called backward multi-trees. Under certain natural restrictions, a backward multi-tree of height $h$ can encode a backward structure of depth $h$. In Lemma~\ref{lm:lower bound proba}, we show that for every backward-structure of size at most $\frac{n}{3}$, the probability that a random structure of size $n$ has this backward structure at depth $h$ is lower-bounded by the probability that the process $\G_h$ produces a multi-tree that encodes this backward structure. 

Thanks to this reduction, it only remains to show that with visible probability the process $\G_h$ produces a multi-tree encoding a backward structure satisfying the {restrictions} on its total size and the number of nodes at depth $h$ required in Proposition~\ref{pro:backward}. A key ingredient in obtaining this result is that the process $\G_h$ is constructed in such a way that if we only consider the number of vertices at each depth and not the structure of the multi-tree, we are in fact studying a Galton-Watson process with offspring distribution $\mathrm{Poi}(2)$. This model is well-studied and we derive the results necessary for our purpose in Subsection~\ref{sec:GW}, using standard techniques.

\subsection{Backward substructure}\label{sec:backward substructure}
Let $\A$ be a $n$-state deterministic transition structure and $p\in[n]$. For any state $x\in[n]$, let $\dback_p(x)$ denote the \emph{backward distance} from $p$ to $x$, that is, the length of a shortest path from $x$ to $p$ in $\A$. By convention, $\dback_p(x)=+\infty$ if $p$ is not reachable from $x$. For $t\geq 0$, the $t$-th \emph{backward layer} of $\A$ from $p$ is the set $L_t(\A,p)$ (simply denoted by $L_t$ or $L_t(\A)$ when the context is clear) consisting of the states at backward distance $t$ from $p$: $L_t(\A,p)=\{x:\dback_p(x)=t\}$.

For any integer $h\geq 0$, we define the \emph{$h$-backward substructure} $\B_{h}$  of $\A$ from $p$ as the incomplete deterministic transition structure obtained by selecting the states at backward distance at most $h$ from $p$ and the transitions that are part of their shortest paths to $p$. Formally, if $\A=(n,\delta)$, then $\B_h=(n,\gamma)$, where $\gamma$ is the partial transition function defined by
$\gamma(x,\alpha)=y$ if and only if $\delta(x,\alpha)=y$ and $\dback_p(x) \leq h$ and $\dback_p(y)= \dback_p(x)-1$. The \emph{backward support} of $\B_h$, denoted $\Support{\B_h}$ is $\cup_{k=0}^h L_k$, the set of states of $\B_h$ which are either the source or target of a transition, or just $\{p\}$ in the special case of $h=0$. When there is no ambiguity on $h$ and $p$, we sometimes write  $\Support{\A}$ instead of $\Support{\B_h}$.

For any $n\geq 1$, any $p\in[n]$ and any $h\geq 0$, denote by $\backT_{n,h}(p)$ the set of all possible $h$-backward substructures from $p$ obtained from a $n$-state deterministic transition structure. 

The template \emph{associated to} $\B_h=(n,\gamma)\in\backT_{n,h}(p)$ is the template $\T=(n,\lambda)$ defined for every $x\in[n]$ and $\alpha\in\Sigma$ by:
\begin{enumerate}
	\item If $\gamma(x,\alpha)$ is defined, then $\lambda(x,\alpha)=\{\gamma(x,\alpha)\}$.
	\item If $\gamma(x,\alpha)$ is undefined and $x\in\Support{\B_h}$, then
	$\lambda(x,\alpha)=[n]\setminus\cup_{k=0}^{\dback_p(x)-1} L_k$.
	\item Otherwise, if $x\notin\Support{\B_h}$, then $\lambda(x,\alpha)=[n]\setminus\cup_{k=0}^{h-1} L_k$.
\end{enumerate}

This template is defined to precisely capture all $n$-state transition structures admitting $\B_h$ as $h$-backward substructure.

\begin{lem}\label{lm:associated template}
If $\T$ is the template associated to $\B_h=(n,\gamma)\in\backT_{n,h}(p)$, then a $n$-state transition structure $\A$ satisfies $\T$ if and only if its $h$-backward substructure is $\B_h$.
\end{lem}

\begin{proof}
Suppose that $\A=(n,\delta)$ satisfies $\T$. A direct induction shows that for all $t\in\{0,\ldots,h\}$ we have $L_t(\A) = L_t(\B_h)$. Moreover, if  $x\in L_t$ for 
$t\in\{1,\ldots,h\}$, then the conditions of Case~(2) in the definition of $\T$ yields that $\delta(x,\alpha)\in L_{t-1}$ if and only if $\gamma(x,\alpha)$ is defined, yielding $\delta(x,\alpha)=\gamma(x,\alpha)$ {as} we must be in  Case~(1). Hence the $h$-backward substructure of $\A$ is $\B_h$. 

Conversely, if the $h$-backward structure of $\A$ is $\B_h$, then for every $(x,\alpha)\in[n]\times\Sigma$:
\begin{enumerate}
	\item If $\gamma(x,\alpha)$ is defined, then $\delta(x,\alpha)=\gamma(x,\alpha)$, thus $\delta(x,\alpha)\in\lambda(x,\alpha)$.
	\item If $\gamma(x,\alpha)$ is undefined and $x\in\Support{\B_h}$, then $\delta(x,\alpha)$ cannot be in $\cup_{k=0}^{d-1} L_k$, where $d=\dback_p(x)$ in $\B_h$, otherwise it would create a shorter path from $x$ to $p$ in $\A$, and $\B_h$ would not be its $h$-backward substructure.
	\item If $x\notin\Support{\B_h}$, using the same argument, $\delta(x,\alpha)$ cannot be in $\cup_{k=0}^{h-1} L_k$.
\end{enumerate} 
This concludes the proof.
\end{proof}
Following the statement of Lemma~\ref{lm:associated template}, we write $\A\satisfies\B_h$ when $\A$ satisfies the template associated with $\B_h$. We further extend the notation to sets $\mathfrak{X}$ of backward substructures: $\A\satisfies\mathfrak{X}$ if $\A$ satisfies at least one element of $\mathfrak{X}$. Lemma~\ref{lm:associated template} also allows us to define $\Support{\T}=\Support{\B_h}$, when $\T$ is the template associated with  $\B_h\in\backT_{n,h}(p)$.

\subsection{Backward multi-trees, trees and processes}\label{sec:backward tree}

We define backward multi-trees which are general trees and identify a subclass called backward trees that naturally encode backward substructures. Using this view, we can consider a simpler stochastic process that produces backward multi-trees instead of backward-substructures. 
In Lemma~\ref{lm:lower bound proba}, we show that this  stochastic process can be used {to give a lower-bound on the probability} that a given backward substructure appears in a random transition structure.

A \emph{$n$-backward multi-tree} rooted at $p$  is a finite general\footnote{In a general tree, the children of each node are not ordered.} {tree whose} root is labeled by $p$ and {whose other nodes} are labeled by pairs $(x,\alpha)\in[n]\times\Sigma$. By a slight abuse of notation, if a node $N$ is labeled by $(x,\alpha)$, we say that $x$ is the state that labels $N$ or that $x$ labels $N$.
The name multi-tree underlines the fact that a state can label several nodes in the tree.

A \emph{$n$-backward tree} is a $n$-backward multi-tree whose nodes are labeled by pairwise distinct states. In particular, its labels are pairwise distinct.  The notions of backward multi-tree and of backward tree are illustrated in Figure~\ref{fig:Gt}, for $n=8$.

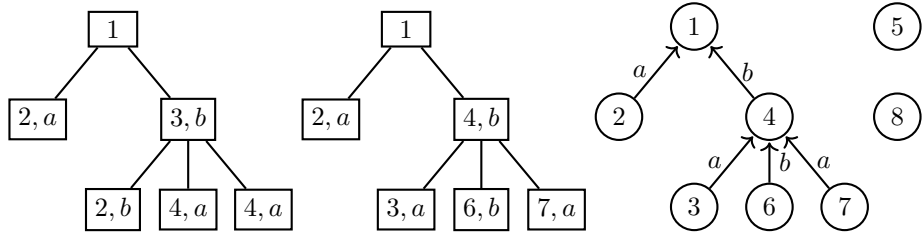
\begin{figure}[h]
\small
\centering
\begin{tikzpicture}[yscale=1.2]
	\node[thick, draw] (p) at (0,0) {$\ 1\ $};

	\node[thick, draw] (p0) at (-1,-1) {$2,a$};
	\node[thick, draw] (p1) at (1,-1) {$3,b$};

	\node[thick, draw] (p10) at (0,-2) {$2,b$};
	\node[thick, draw] (p11) at (1,-2) {$4,a$};
	\node[thick, draw] (p12) at (2,-2) {$4,a$};
	
	\draw[thick] (p) -- (p0);
	\draw[thick] (p) -- (p1);
	\draw[thick] (p1) -- (p10);
	\draw[thick] (p1) -- (p11);
	\draw[thick] (p1) -- (p12);
\end{tikzpicture}
\begin{tikzpicture}[yscale=1.2]
	\node[thick, draw] (p) at (0,0) {$\ 1\ $};

	\node[thick, draw] (p0) at (-1,-1) {$2,a$};
	\node[thick, draw] (p1) at (1,-1) {$4,b$};

	\node[thick, draw] (p10) at (0,-2) {$3,a$};
	\node[thick, draw] (p11) at (1,-2) {$6,b$};
	\node[thick, draw] (p12) at (2,-2) {$7,a$};
	
	\draw[thick] (p) -- (p0);
	\draw[thick] (p) -- (p1);
	\draw[thick] (p1) -- (p10);
	\draw[thick] (p1) -- (p11);
	\draw[thick] (p1) -- (p12);
\end{tikzpicture}
\begin{tikzpicture}[yscale=1.2]
	\node[thick, draw,circle] (p) at (0,0) {$1$};

	\node[thick, draw,circle] (p0) at (-1,-1) {$2$};
	\node[thick, draw,circle] (p1) at (1,-1) {$4$};

	\node[thick, draw,circle] (p10) at (0,-2) {$3$};
	\node[thick, draw,circle] (p11) at (1,-2) {$6$};
	\node[thick, draw,circle] (p12) at (2,-2) {$7$};
	
	\draw (p0) edge[thick, ->, shorten >= .5pt] node[left]{$a$} (p);
	\draw (p1) edge[thick, ->, shorten >= .5pt] node[right]{$b$} (p);

	\draw (p10) edge[thick, ->, shorten >= .5pt] node[left]{$a$} (p1);
	\draw (p11) edge[thick, ->, shorten >= .5pt] node[right]{$b$} (p1);
	\draw (p12) edge[thick, ->, shorten >= .5pt] node[right]{$a$} (p1);

    \node[thick, draw,circle] (p) at (2.7,0) {$5$};
    \node[thick, draw,circle] (p) at (2.7,-1) {$8$};

\end{tikzpicture}
\caption{On the left, a backward multi-tree which is not a backward tree because, for instance, the state~2 labels two nodes. In the center, a backward multi-tree which is a tree with its corresponding backward substructure on the right. \label{fig:Gt}}
\end{figure}

A $n$-backward tree $T$ of depth at most $h$ can be very simply transformed into a $h$-backward substructure with $n$ states, by starting with no transition then changing every child$\rightarrow$parent relation $(x,\alpha)\rightarrow(y,\beta)$ into the transition $x\xrightarrow{\alpha}y$, and $(x,\alpha)\rightarrow p$ into
$x\xrightarrow{\alpha}p$: each letter is moved from its node's label to the edge linking it to its parent, and the unused states are added with no outgoing transition. We denote by $\Lambda$ this map from a $n$-backward tree to a $n$-state backward substructure.

We say that a $n$-backward multi-tree $T$ \emph{matches} a backward substructure $\B$ when $T$ is a tree and $\Lambda(T)=\B$.
Let $\treeT_n$ denote the set of $n$-backward multi-trees  that are trees, and let $\treeT_{n,h}(p)$ be the set of backward trees that match a $h$-backward substructure at $p$:
\[
\treeT_{n,h}(p):=\{T\in\treeT_n: \Lambda(T)\in\backT_{n,h}(p)\}.
\]
Observe that a backward substructure can be matched by a backward tree if and only if every state in its support has exactly one outgoing transition, except $p$ which has none. We say that such a backward substructure is a \emph{tree}.

For $n\geq 1$ and $p\in[n]$, we define a stochastic process $(\G_t)_{t\geq 0}$, called the \emph{backward multi-tree process} of parameters $n$ and $p$, which produces random $n$-backward multi-trees as follows. At time $t=0$, the tree consists of a unique root node labeled $p$, which is at depth $0$. For $t\geq 0$,
the tree $\G_{t+1}$ is built from $\G_t$ in the following way. For every node $N$ at depth $t$ in $\G_t$, we draw independent Poisson random variables $X_{x,\alpha}$ of parameter $\frac1n$, one for every $x\in[n]$ and for every $\alpha\in\Sigma$. If $X_{x,\alpha}>0$, then we add $X_{x,\alpha}$ children labeled $(x,\alpha)$ to $N$, at depth $t+1$. 

The following lemma justifies the introduction of the multi-tree process $(\G_t)_{t\geq 0}$ as we will use it to lift some visible properties of a backward tree generated by this process to the visible properties of the backward substructure of a random transition structures. 

\begin{lem}\label{lm:lower bound proba}
Let $\B\in\backT_{n,h}(p)$ %
such that $|\Support{\B}| < \frac{n}3$ and the $h$-th layer of $\B$ is not empty. 
Let $\A_n$ denote a uniform random transition structure with $n$ states, and $(\G_t)_{t\geq 0}$ be the backward tree process of parameters $n$ and $p$.  Then $\proba(\A_n\satisfies\B) \geq \proba(\G_h\text{ matches }\B)$.
\end{lem}

\begin{proof}
If $\B$ is not a tree then $\proba(\G_h\text{ matches }\B)=0$ and the inequality trivially holds. So we assume that $\B$ is a tree for the remainder of the proof.

Let $\T$ denote the template associated with $\B$.
We proceed by induction on the depth $h$ of $\B=(n,\gamma)$. The claim trivially holds if $h=0$. 
For the induction step, assume the claim is true for depth $h$ and consider a backward substructure $\B$ of depth $h+1$. 
Let $\hat\B$ be the associated backward substructure of depth $h$, which just consists in removing the outgoing transitions from the states of $L_{h+1}(\B)$. Observe that $\hat\B$ is also a tree, thus we have
$\proba(\G_h\text{ matches }\hat\B)>0$.

Let $\hat\T$ denote the template  associated with $\hat\B$. As
 $|\Support{\hat\B}| \leq |\Support{\B}| \leq \frac{n}3$, by induction hypothesis  we have $\proba(\A_n\satisfies\hat\T) \geq \proba(\G_h\text{ matches }\hat\B)$. 

 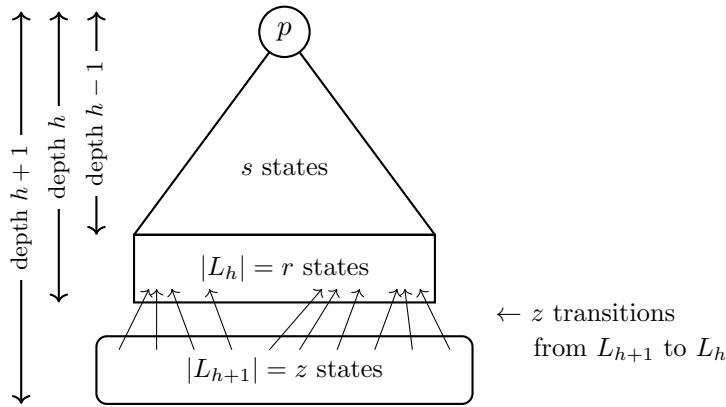
\begin{figure}[h]
\centering
\begin{tikzpicture}[yscale=.9]

\draw[thick] (0,0) -- (-2,-3);
\draw[thick] (0,0) -- (2,-3);
\draw[thick] (-2,-3) -- (2,-3);
\node[draw,thick,circle,fill=white] at (0,0) {$p$};
\draw[thick,<->] (-2.5,0.3) -- node[rotate=90,fill=white]{\footnotesize depth $h-1$} (-2.5,-3);
\draw[thick,<->] (-3,0.3) -- node[rotate=90,fill=white]{\footnotesize depth $h$} (-3,-4);
\draw[thick,<->] (-3.5,0.3) -- node[rotate=90,fill=white]{\footnotesize depth $h+1$} (-3.5,-5.5);

\draw[thick] (-2,-3) rectangle (2,-4);

\draw[thick, rounded corners] (-2.5,-4.5) rectangle (2.5,-5.5);

\node at (0,-2) {\small $s$ states};
\node at (0,-3.5) {\small $|L_h|=r$ states};
\node at (0,-5) {\small $|L_{h+1}|=z$ states};

\draw[->] (-2.2,-4.7) -- (-1.8,-3.8);
\draw[->] (-1.7,-4.7) -- (-1.7,-3.8);
\draw[->] (-1.2,-4.7) -- (-1.5,-3.8);
\draw[->] (-.7,-4.7) -- (-1,-3.8);
\draw[->] (-.2,-4.7) -- (.5,-3.8);
\draw[->] (.2,-4.7) -- (.7,-3.8);
\draw[->] (.7,-4.7) -- (1,-3.8);
\draw[->] (1.2,-4.7) -- (1.5,-3.8);
\draw[->] (1.7,-4.7) -- (1.6,-3.8);
\draw[->] (2.2,-4.7) -- (1.8,-3.8);

\node at (4,-4.15) {\small $\leftarrow$ $z$ transitions};
\node at (4.6,-4.7) {\small from $L_{h+1}$ to $L_h$};

\end{tikzpicture}
\caption{An illustration for the values of $r$, $s$ and $z$ in the proof of Lemma~\ref{lm:lower bound proba}.\label{fig:proof lower bound}}
\end{figure}

Let $r=|L_{h}(\B)|=|L_h(\hat\B)|$, $s=|\cup_{k=0}^{h-1} L_k|$ and $z=|L_{h+1}(\B)|$. Observe that since $\B$ is a tree,  $z$  is also the number of transitions from $L_{h+1}$ to $L_h$ in $\B$. They are the transitions we removed to build $\hat\B$.

A transition structure $\A=(n,\delta)$ that satisfies $\hat\T=(n,\hat\lambda)$ also satisfies $\T$ if and only if for every $x\notin\Support{\hat\B}$ and every $\alpha\in\Sigma$ we have  (i) if $\gamma(x,\alpha)$ is defined in $\B$ then $\delta(x,\alpha)=\gamma(x,\alpha)$, and (ii) if $\gamma(x,\alpha)$ is undefined in $\B$ then $\delta(x,\alpha)\notin\Support{\hat\B}$. 
Since $\A\satisfies\hat\T$, for every $x\notin\Support{\hat\B}$ and every $\alpha\in\Sigma$, $\hat\lambda(x,\alpha)=[n]\setminus \cup_{k=0}^{h-1} L_k$, hence $|\hat\lambda(x,\alpha)|=n-s$. 
Out of these $n-s$ possible targets, only one is possible for the $z$ transitions of Case (i) and $n-r-s$ are possible for the $2(n-r-s)-z$ transitions of Case (ii). Since the other transitions have the same constraints in $\T$ and in $\hat\T$, by Lemma~\ref{lm:proba template} we have 
\[
\proba(\A_n\satisfies\T \mid \A_n\satisfies\hat\T) = \frac1{(n-s)^z}\times\left(\frac{n-r-s}{n-s}\right)^{2(n-s-r)-z},
\] 
We rewrite this probability the following way:
\[
\proba(\A_n\satisfies\T \mid \A_n\satisfies\hat\T) = \frac1{(n-r-s)^z}\times\left(1-\frac{r}{n-s}\right)^{2(n-s-r)}.
\] 
On the other hand, by construction of the process and since $\proba(\mathrm{Poi}(\frac1n)=0)=e^{-1/n}$ and $\proba(\mathrm{Poi}(\frac1n)=1)=\frac1n e^{-1/n}$, considering the $2nr$ possible transitions that ends in one of the $r$ vertices of the last layer in $\hat\B$:
\[
\proba(\G_{h+1}\text{ matches }\B\mid \G_{h}\text{ matches }\hat\B) = \left(\frac1n e^{-1/n}\right)^z \times\left(e^{-1/n}\right)^{2nr-z} = \frac{e^{-2r}}{n^z}.
\]
Let $R$ be the ratio of these two probabilities, we have
\[
R = \frac{\proba(\A_n\satisfies\T \mid \A_n\satisfies\hat\T)}{\proba(\G_{h+1}\text{ matches }\B\mid \G_{h}\text{ matches }\hat\B)}
\geq \left(1-\frac{r}{n-s}\right)^{2(n-s-r)} e^{2r}.
\]
As $-\log(1-x) \geq x$ for $x\in(0,1)$, we have
\[
\log R  \geq 2(n-s)\log\left(1-\frac{r}{n-s}\right) +\frac{2r^2}{n-s} + 2r.
\]
A basic study of the function $x\mapsto \log(1-x) + x + \frac34 x^2$ yields that
for all $x\in[0,\frac13]$, we have $\log(1-x) \geq -x-\frac34 x^2$. Moreover, $r+\frac{s}3\leq r+s = |\Support{\B}|\leq \frac{n}3$, hence $r\leq \frac{n-s}3$. Therefore, we have
\[
\log R  \geq -2(n-s)\left(\frac{r}{n-s}+\frac{3r^2}{4(n-s)^2}\right)+\frac{2r^2}{n-s} + 2r = \frac{r^2}{2(n-s)} \geq 0.
\]
This yields that $\proba(\A_n\satisfies\T \mid \A_n\satisfies\hat\T) \geq \proba(\G_{h+1}\text{ matches }\B\mid \G_{h}\text{ matches }\hat\B)$. Therefore,
\begin{align*}
\proba\left(\A_n\satisfies\T\right) 
& = \proba\left(\A_n\satisfies\T \text{ and } \A_n\satisfies\hat\T\right) \\
& = \proba\left(\A_n\satisfies\T \mid \A_n\satisfies\hat\T\right) \proba\left(\A_n\satisfies\hat\T\right) \\
& \geq \proba\left(\G_{h+1}\text{ matches }\B\mid \G_{h}\text{ matches }\hat\B\right) \,\proba\left(\G_{h}\text{ matches }\hat\B\right)\\
& = \proba\left(\G_{h+1}\text{ matches }\B\text{ and } \G_{h}\text{ matches }\hat\B\right)\\
& = \proba\left(\G_{h+1}\text{ matches }\B\right),
\end{align*}
concluding the proof by induction.
\end{proof}

\subsection{Results in Galton-Watson processes of offspring distribution \texorpdfstring{$\textrm{Poi}(2)$}{Poi(2)}}\label{sec:GW}

Let $(\G_t)_{t\geq 0}$ be a backward multi-tree process of parameters $n$ and $p$. To each $\G_t$ we associate the quantity $Z_t$, defined as its number of nodes at depth $t$. 
{Since the sum of $|\Sigma|\cdot n$ independent Poisson random variables of parameter $\frac1n$ is a Poisson random variable of parameter $|\Sigma|$ and $|\Sigma|=2$}, each node at depth $t$ in $\G_t$ gives birth to $\textrm{Poi}(2)$ children at time $t+1$. Hence $(Z_t)_{t\geq 0}$ is exactly a Galton-Watson process of offspring distribution $\textrm{Poi}(2)$, which is a well studied branching process~\cite{Harris}.
We tailor classical results on Galton-Watson processes to our needs as follows.
\begin{thm}\label{thm:GW}
Let $(Z_t)_{t\geq 0}$ be  a Galton-Watson process of offspring distribution $\mathrm{Poi}(2)$. There exist positive real numbers $c, c_1, c_2, c_3$ and a positive integer $t_0$ such that, 
\[
\proba\left(\forall t\geq t_0,\ c_1 2^t\leq Z_t\leq c_2 2^t\text{ and }\sum_{k=0}^{t-1} Z_k \leq c_3 2^t\right) \geq c.
\]
\end{thm}

\begin{proof}

As the expected value of the offspring law is 2 and its variance is finite,  the sequence of random variables $(W_t)_{t \geq 0}$ defined by $W_t=2^{-t}Z_t$ almost surely converges to a random variable $W$ with $\mathbb{E}[W]=1$ (see \cite[Theorem 8.1, p.13]{Harris}). Since $\mathbb{E}[W]>0$, there must exist a constant ${d}>0$ such that $\proba(|W-{d}|<\frac{{d}}{4})>0$: we can cover $\mathbb{R}_{\geq 0}$ with a countable number  
of intervals $(\frac34{d},\frac54{d})$, and $W$ is in at least one of them with positive probability.

For $t_0 \geq 1$ and $x>0$, we consider the following events:
\begin{itemize}
    \item event $A_{t_0}$ : $\bigcap_{t=t_0}^{\infty} | W_t - W | < \frac{{d}}{4}$,
    \item event $B$ : $|W - {d} |<\frac{{d}}{4}$,
    \item event $C_{t_0,x}$ : 
    $\sum_{t=0}^{t_0-1} Z_t \leq x (2^{t_0}-1)$.        
\end{itemize}
Our choice of ${d}$ yields that ${c}:=\proba(B)$ is positive. As  $(W_t)_{t \geq 0}$  converges almost surely to $W$, $\lim_{t_0 \rightarrow \infty} \proba(A_{t_0}) = 1$. So we can fix a value for $t_0$ such that $\proba(A_{t_0})\geq 1-\frac{{c}}{4}$.
For this $t_0$, $\lim_{x \rightarrow \infty} \proba(C_{t_0,x})=1$, so we can choose some value for $x$ such that $\proba(C_{t_0,x}) \geq 1-\frac{{c}}{4}$. For these choices of $t_0$ and $x$, this yields, by the union bound, using $\overline{X}$ for the complement of $X$:
\[
\proba\left(A_{t_0} \cap B \cap C_{t_0,x}\right) \geq
1 -  \proba\left(\overline{A}_{t_0}\right) - \proba\left(\overline{B}\right)-\proba\left(\overline{C}_{t_0,x}\right)\geq \frac{{c}}{2} > 0.
\]
We finalize the proof as follows. Assume that the events
$A_{t_0}$, $B$ and $C_{t_0,x}$ occur simultaneously.
For all $t \geq t_0$, we have $|W_t-{d}| \leq |W_t-W| + |W-{d}| \leq \frac{{d}}{4} + \frac{{d}}{4} \leq \frac{{d}}{2}$. Hence, for all $t \geq t_0$, $ \frac{{d}}{2} 2^t \leq Z_t \leq \frac{3{d}}{2} 2^t$, and we can take $c_1 = \frac{{d}}{2}$ and $c_2 = \frac{3{d}}{2}$.
Furthermore, for $t > t_0$, we have:
\[
\sum_{k=0}^{t-1} Z_k = \sum_{k=0}^{t_0-1} Z_k + \sum_{k=t_0}^{t-1} Z_k \leq
x(2^{t_0}-1)+ \frac{3 {d}}{2} \sum_{k=t_0}^{t-1}  2^k,
\]
which is at most $c_32^t$ for $c_3=\max\{x,3{d}/2\}$, concluding the proof.
\end{proof}

\begin{lem}\label{lm:GW markov}
Let $(Z_t)_{t\geq 0}$ be  a Galton-Watson process of offspring distribution $\mathrm{Poi}(2)$.  The probability that $\sum_{i=0}^{t}Z_{i} \geq 2^{3t/2}$ tends to $0$ as $t$ tends to infinity.
\end{lem}

\begin{proof}
For all $t \geq 0$, we have $\mathbb{E}[Z_t]=2^t$. By linearity of expectation, $\mathbb{E}[\sum_{i=0}^{t}Z_{i}]=2^{t+1}-1$. Hence by Markov inequality, for all $i\geq 0$, $\proba(\sum_{i=0}^{t}Z_{i} \geq 2^{3t/2}) \leq \frac{(2^{t+1}-1)}{2^{3t/2}}$ which tends to $0$ as $t$ tends to infinity, concluding the proof.  
\end{proof}

\subsection{Main result on backward substructures}

The aim of this section is to conclude our work on backward substructures, by proving Proposition~\ref{pro:backward} stated below.

Let $\mathfrak{D}$ denote the set of backward multi-trees in which no node has two children with the same label. For instance in Figure~\ref{fig:Gt}, the multi-tree on the left is not in $\mathfrak{D}$ as $(3,b)$ has two children with same label $(4,a)$. We first show that with high probability the multi-tree produced by the process is in $\mathfrak{D}$.

\begin{lem}\label{lm:no duplicates}
Let $(\G_t)_{t\geq 0}$ be the backward multi-tree process of parameters $n$ and $p$. With high probability, $\G_h\in\mathfrak{D}$ for $h:=\lceil\log_2\sqrt{n}\,\rceil$.
\end{lem}

\begin{proof}

A node has children with the same label if and only if at least one of the $\mathrm{Poi}(1/n)$ used to construct its children values at least two.

The probability that such a random variable is at least two is $1-e^{-1/n}(1+1/n) = 1/n^2 + o(1/n^2)$. Thus if we consider a sequence of $o(n^2)$ such random variables, the probability that at least one of them reaches two or more is  $o(1)$ by the union bound. 

Let $(\xi_i)_{i\geq 1}$ be a sequence of i.i.d. $\mathrm{Poi}(1/n)$ random variables used for the construction of $\G_h$. By the union bound,   as $\proba\left(|\G_h|\geq 2^{3h/2}\text{ and }\G_h\notin\mathfrak{D}\right)\leq \proba\left(|\G_h|\geq 2^{3h/2}\right)$ we have
\[
\proba\left(\G_h\notin \mathfrak{D}\right) \leq
\proba\left(|\G_h|\geq 2^{3h/2}\right) +
\proba\left(\exists i\in\{1,\ldots,2n\,\lfloor 2^{3h/2}\rfloor\}, \xi_i\geq 2\right).
\]
The first probability tends to $0$ by Lemma~\ref{lm:GW markov}. The second probability also tends to $0$ by the above remark and the fact that $n\lfloor2^{3h/2}\rfloor = \Theta(n^{7/4})$.
This concludes the proof.
\end{proof}

For a multi-tree in $\mathfrak{D}$, we define a notion of shape which intuitively corresponds to drawing the multi-tree in the plane with the children ordered lexicographically and then in removing all the labels. By definition of $\mathfrak{D}$,   all children of a node have different labels, hence  this  embedding is uniquely defined.
The proof of Proposition~\ref{pro:backward} below consists in showing that for a fixed shape $S$, given it has shape $S$, the probability that $\G_h$ is a tree  can be bounded from below by some constant $\tau$ independently of the shape $S$. We will then conclude by the law of total probabilities.

\begin{prop}\label{pro:backward}
There exists $c_\mathrm{B},d_1,d_2>0$ such that, for any $n$ sufficiently large and any $p\in[n]$,
the $h$-backward substructure $\B$ from $p$ of depth $h=\lceil\log_2  \sqrt{n}\,\rceil$ of a uniform random $n$-states transition structure is such that $|\Support{\B}|\leq d_1\sqrt{n}$ and $|L_h(\B)|\geq d_2\sqrt{n}$  with probability at least $c_\mathrm{B}$.
\end{prop}

\begin{proof}

We choose some values for $c,c_1,c_2,c_3$ {and $t_0$} that work with Theorem~\ref{thm:GW}. We set $d_1=2(c_2+c_3)$ and $d_2=c_1$.
Let $(\G_t)_{t\geq 0}$ be a backward multi-tree process of parameters $n$ and $p$, and let $(Z_t)_{t\geq 0}$ denote the number of nodes at depth $t$ in $\G_t$ (and therefore in all $\G_{t'}$ for $t'\geq t$). 

Let $\mathfrak{G}_h$ denote the set defined by:
\[
\mathfrak{G}_h = \left\{\G\in\mathfrak{D}: c_12^h\leq Z_h(\G) \leq c_2 2^h \text{ and }\sum_{k=0}^{h-1}Z_k(\G) \leq c_3 2^h\right\}.
\]
For $n$ sufficiently large, $h\geq t_0$ and  Theorem~\ref{thm:GW} applies. Hence, together with Lemma~\ref{lm:no duplicates}, for $n$ sufficiently large we have 
\begin{equation}\label{eq:proba Gh}
\proba(\G_h\in \mathfrak{G}_h ) \geq \frac{c}2,
\end{equation}
as the intersection of a property with visible probability and a property with high probability.

We start by showing that there is a visible probability that $\G_h$ belongs to $\mathfrak{G}_h$ and is a tree. To do so, we introduce the notion of shape of a backward multi-tree and bound from below the probability at fixed shape.

To any $\G\in\mathfrak{D}$ we associate its \emph{shape} $s(\G)$, a tuple of non-negative integers which is  inductively defined as follows:
\begin{itemize}
	\item if $\G$ is reduced to a root node, then $s(\G)=(0)$,
	\item otherwise, if $(\mathcal{G}_i)_{i=1,\ldots,k}$ denote its $k$ children ordered by their root labels (using the lexicographic order), then
	$s(\G) = (k)\otimes s(\G_1)\otimes s(\G_2)\otimes\cdots\otimes s(\G_k)$, where $\otimes$ denote the concatenation of tuples: $(x_1,\ldots,x_m)\otimes (y_1,\ldots,y_r)=(x_1,\ldots,x_m,y_1,\ldots,y_r)$.
\end{itemize}
So we compute $s(\G)$ by producing the sequence of the number of children in a depth-first traversal of $\G$, where the children of a node are taken in lexicographic order of their labels. There is no ambiguity as in an element of $\mathfrak{D}$, two siblings always have different labels. If $\G\notin\mathfrak{D}$, we set $s(\G)=\bot$ to indicate that no shape is defined.

The label sequence of $\G\in\mathfrak{D}$, denoted by $\ell(\G)$ is computed  as $s(\G)$, except that we collect the  labels of the nodes instead of their numbers of children.
One can readily verify that if both $s(\G)$ and $\ell(\G)$ are given, there is a unique element of $\mathfrak{D}$ that matches them, which is $\G$: informally $s(\G)$ encodes the tree structure, and the nodes are labeled with $\ell(\G)$ using a depth-first traversal.
Let $\shape_h=\{s(\G):\G\in\mathfrak{G}_h\}$ be the set of shapes of
elements of $\mathfrak{G}_h$. Observe that if $s(\G)\in\shape_h$, 
then by definition $\G\in\mathfrak{D}$  (otherwise $s(\G)=\bot$); moreover, the shape of $\G$ determines the size of its layers, and therefore $\G\in\mathfrak{G}_h$. Hence $s(\G)\in\shape_h$ if and only if $\G\in\mathfrak{G}_h$.

The probability that $\G_h=\G$ is the same for any $\G$ having the same shape $\cS\in\shape_h$. Indeed, by direct induction on the steps of the depth-first traversal, each time we generate the $k$ children of the current node, the value of $k$ being given by $s(\G)$, we draw $2n$ independent $\textrm{Poi}(1/n)$ out of which exactly $k$ must value 1 and $2n-k$ must value $0$. This happens with probability that only depends on $k$, hence only on $s(\G)$.

So we can compute the probability that $\G_h$ is a backward tree, given its shape $\cS$, by a counting argument. Indeed, to construct a backward tree with shape $\cS=(s_1,\ldots,s_t)$, we first need to choose the states labeling the $s_1$ children of the root (which is always labeled by $p$) and for each child assign a letter in $\Sigma$. This represents  $\binom{n-1}{s_1}2^{s_1}$ choices. Then we need to choose the $s_2$ states labeling the $s_2$ children of the first child of the root. These $s_2$ states must be chosen in amongst the $n-1-s_1$ not used so far. For each of these states, which is uniquely identified by its label,  we must choose a letter in $\Sigma$ for a total number of  $\binom{n-1-s_1}{s_2}2^{s_2}$ choices (observe that the formula is also correct if $s_2=0$). By repeating this argument in the form of a direct induction, we have that the number of ways to build a backward tree of shape $\cS=(s_1,\ldots,s_t)$ is 
\[
T_n(\cS):=\binom{n-1}{s_1}2^{s_1}\times \binom{n-1-s_1}{s_2}2^{s_2}\times \cdots
\times \binom{n-1-\sum_{i=1}^{t-1}s_i}{s_t}2^{s_t}.
\]
On the other hand, the number of way of labeling $\cS$ to obtain an element of $\mathfrak{D}$ is
\[
D_n(\cS):=\binom{2n}{s_1}\times \binom{2n}{s_2}\times \cdots
\times \binom{2n}{s_t},
\]
as the only constraint is that the children of a node have different labels, amongst the $2n$ possible ones. {Recalling} that $\mathfrak{T}_n$ denotes the set of $n$-backward trees, we just established that $\proba\left(\G_h\in\mathfrak{T}_n\mid s(\G_h)=\cS\right)=
\frac{T_n(\cS)}{D_n(\cS)}$. We can easily bound from below as follows, using $m=1+\sum_{i=1}^t s_i$ to denote the number of nodes:
\[
T_n(\cS) \geq 2^{m-1} \frac{\left(n-m\right)^{m-1}}{\prod_{i=1}^t s_i!}
\text{ and }
D_n(\cS) \leq \frac{(2n)^{m-1}}{\prod_{i=1}^t s_i!}.
\]
Therefore, we have
\[
\proba\left(\G_h\in\mathfrak{T}_n\mid s(\G_h)=\cS\right)
\geq \left(1-\frac{m}{n}\right)^{m}.
\]
In particular, if $\cS$ has at most $d_1\sqrt{n}$ nodes, then for $n$ sufficiently large we have
\begin{equation}\label{eq:lower bound T}
\proba\left(\G_h\in\mathfrak{T}_n\mid s(\G_h)=\cS\right)\geq \frac12\exp(-d_1^2), 
\end{equation}
as $\lim_n(1-d_1/\sqrt{n})^{d_1\sqrt{n}}=e^{-d_1^2}$.

We can now finalize the proof by bounding from below the probability that $\G_h$ is a $n$-backward tree and in $\mathfrak{G}_h$ as follows, for $n$ sufficiently large. We first partition the elements of $\mathfrak{G}_h$ according to their shape:
\begin{align*}
\proba\left(\G_h\in \mathfrak{G}_h\cap\mathfrak{T}_n\right)  &= \sum_{\cS\in\shape_h}
\proba\left(s(\G_h)=\cS\text{ and }\G_h\in\mathfrak{T}_n\right)\\
& =  \sum_{\cS\in\shape_h}
\proba\left(\G_h\in\mathfrak{T}_n\mid s(\G_h)=\cS\right)
\proba\left(s(\G_h)=\cS\right).
\end{align*}
As $\shape_h$ is the set of shapes of the elements of $\mathfrak{G}_h$, all the $\G$ such that $s(\G)\in\shape_h$ has at most $(c_2+c_3)2^{h}\leq 2(c_2+c_3)\sqrt{n}=d_1\sqrt{n}$ nodes, Equation~\eqref{eq:lower bound T} applies and yields:
\[
\proba\left(\G_h\in \mathfrak{G}_h\cap\mathfrak{T}_n\right) \geq 
\frac{e^{-d_1^2}}2\,\sum_{\cS\in\shape_h}\proba\left(s(\G_h)=\cS\right) =
\frac{e^{-d_1^2}}2 \, \proba\left(\G_h\in\mathfrak{G}_h\right).
\]
By Equation~\eqref{eq:proba Gh}, if we set $c_\mathrm{B}=\frac{ce^{-d_1^2}}4$, then for $n$ sufficiently large $\G_h\in \mathfrak{G}_h\cap \mathfrak{T}_n$ with probability at least $c_\mathrm{B}$.

Recall that for $\G\in\mathfrak{T}$, $\Lambda(\G)$ denote the associated backward substructure. Let $\mathfrak{B}_h\subseteq\backT_{n,h}(p)$ denote the set of backward substructures $\B$ such that $|\Support{\B}|\leq d_1\sqrt{n}$ and 
$|L_h(\B)|\geq d_2\sqrt{n}$.
As every structure in $\Lambda(\mathfrak{G}_h\cap \mathfrak{T}_n)$ is in $\mathfrak{B}_h$, we have
\[
\proba\left(\A_n\satisfies\mathfrak{B}_h\right) \geq 
\proba\left(\A_n\satisfies\Lambda(\mathfrak{G}_h\cap \mathfrak{T}_n)\right)
= \sum_{\G\in \mathfrak{G}_h\cap \mathfrak{T}_n} \proba\left(\A_n\satisfies\Lambda(\G)\right).
\]
We used the fact that, by Lemma~\ref{lm:associated template}, two different elements  of
$\mathfrak{G}_h\cap \mathfrak{T}_n$ cannot be the $h$-backward substructures of the same transition structure.

Moreover, by Lemma~\ref{lm:lower bound proba}, for $\G\in\mathfrak{T}_n$ we have
\[
\proba\left(\A_n\satisfies\Lambda(\G)\right)\geq \proba\left(\G_h\text{ matches }\Lambda(\G) \right) = \proba\left(\G_h=\G\right).
\]
Therefore,
\[
\proba\left(\A_n\satisfies\mathfrak{B}_h\right) \geq 
\sum_{\G\in \mathfrak{G}_h\cap \mathfrak{T}_n} \proba\left(\G_h=\G\right)
= \proba\left(\G_h\in \mathfrak{G}_h\cap \mathfrak{T}_n\right) \geq c_\mathrm{B}.
\]
This concludes the proof.
\end{proof}

For any $n\geq 1$ and $p\in [n]$, let $\tempB_n(p)$ (or simply $\tempB_n$ when there is no ambiguity on $p$) denote the set of templates associated with the $\lceil\log_2  \sqrt{n}\,\rceil$-backward substructures $\B$ from $p$ of the $n$-state transition structures such that $|\Support{\B}|\leq d_1\sqrt{n}$ and $|L_h(\B)|\geq d_2\sqrt{n}$.

\section{Forward tree and short cycle around \texorpdfstring{$p$}{p}}\label{sec:forward}

In this section, we continue the construction started in Section~\ref{sec:backward} by performing a forward exploration starting from $\delta(p,a)$ in a random transition structure that satisfies $\tempB_n(p)$.

More precisely, let $\T\in\tempB_n(p)$ of depth $h=\lceil \log_2\sqrt{n}\,\rceil$. For a transition structure $\A=(n,\delta)$ that satisfies $\T$, 
we perform a breadth-first traversal of $\A$ from $r:=\delta(p,a)$: the states are discovered in order $r_0=\delta(r,\varepsilon)$,  $r_1=\delta(r,a)$, $r_2=\delta(r,b)$, $r_3=\delta(r,aa)$, $r_4=\delta(r,ab)$, \ldots, where the words are taken in length-lexicographic order. This process halts as soon as one of the following events happens:
\begin{enumerate}
\item $r_i\in L_h$, where $L_h$ is the $h$-th layer of $\T$, i.e. the states $x$  of $\A$ such that $\dback_p(x)=h$; \label{forward process success}
\item $r_i\in \Support{\B_h}\setminus L_h$, which is only possible for $r_0=\delta(p,a)$, as $\A$ satisfies $\T$; \label{forward process r0}
\item $r_i=r_j$ with $j<i$, i.e. we have a collision with a state already seen during the process; \label{forward process collision}
\item $i=\lceil\sqrt{n}\,\rceil$, we force the process to halt after $\sqrt{n}$ steps.
\label{forward process too long}
\end{enumerate}
We call this  process the \emph{forward process} from $p$, and say that it is a \emph{success} if it halts because Case~\eqref{forward process success} is triggered and it is a \emph{failure} otherwise. It is deterministic for a given $\A$, and we will consider its probability of success for  random transition structures.

\begin{lem}\label{lm:forward process}
There exists a constant $c_{\mathrm{F}}>0$ such that, for $n$ sufficiently large,  for any $p\in[n]$, and for any $\T\in\tempB_n(p)$,  the forward process from $p$ of a uniform random $n$-state transition structure conditioned to satisfy $\T$ is a success with probability at least $c_{\mathrm{F}}$.
\end{lem}

\begin{proof}
Let $\B_h$ be the $h$-backward substructure from $p$ associated to $\T$.
By Lemma~\ref{lm:proba template}, in a random transition structure $\A$ that satisfies $\T$, $\delta(p,a)$ is a uniform random element of $[n]$, and for every $x\notin\Support{\T}$  and every $\alpha\in\Sigma$, $\delta(x,\alpha)$ is a uniform random element of $S_\T:=L_h(\B_h)\cup ([n]\setminus\Support{\B_h})$, every choice being made independently.  

We simulate the evolution of the process by generating the target of the transitions one by one  according to the restrictions induced by $\T$, when needed, in the order of the traversal. 

First, we want to upper bound the probability that the process fails because we reach the threshold of $\sqrt{n}$ steps, i.e. it halts because of Case~\eqref{forward process too long}. As the first transition considered is $\delta(p,a)$, and the process immediately fails when $\delta(p,a)\in\Support{\B_H}\setminus L_h$, it continues with probability $1-\frac{|\Support{\B_h}\setminus L_h|}{n}$. 

For $1\leq i\leq \lceil\sqrt{n}\,\rceil-1$, if $r_0$, \ldots, $r_{i-1}$ avoided the halting conditions, 
then $r_{i}$, which is a uniform element of $S_\T$, also avoids the halting conditions with probability $1-\frac{i+|L_h|}{|S_\T|}$, as it must be different from the previous $r_j$'s and not belong to $L_h$. By independence, it yields that the probability $\pi_n$ that the process halts because of Case~\eqref{forward process too long} is
\[
\pi_n = \left(1-\frac{|\Support{\B_h}\setminus L_h|}{n}\right)\prod_{i=1}^{\lceil\sqrt{n}\,\rceil-1}
\left(1-\frac{i+|L_h|}{|S_\T|}\right) \leq
\left(1-\frac{|L_h|}{n}\right)^{\lceil\sqrt{n}\,\rceil-1}.
\]
As $\T\in\tempB_n(p)$, we have $|L_h|\geq d_2\sqrt{n}$ and thus
$\pi_n\leq (1-d_2/\sqrt{n})^{\sqrt{n}-1}$. Moreover, since $\lim_n(1-d_2/\sqrt{n})^{\sqrt{n}-1} = \exp(-d_2)$, for $n$ large enough we have $\pi_n \leq \kappa$ by choosing, for instance, $\kappa:=\frac12(1+\exp(-d_2))\in(0,1)$.

Secondly, if we condition the process to halt at some step $i<\lceil\sqrt{n}\rceil$  
then $r_i$ is chosen uniformly at random
in $\{r_0,\ldots,r_{i-1}\}\cup L_h$. Hence,
given it halts at step $i$, the probability of success is 
\[
\frac{|L_h|}{i+|L_h|} \geq \frac{d_2\sqrt{n}}{\sqrt{n}+d_1\sqrt{n}} = \frac{d_2}{1+d_1}.
\]
Finally, since this lower bound does not depend on $i$, the probability that it halts because of Case~\eqref{forward process success}, i.e. the process is a success,  is at least $\frac{d_2}{1+d_1}$ times the probability that the process halts before $\lceil \sqrt{n}\rceil$ steps. Hence, it is at least $c_\mathrm{F}=\frac{d_2}{1+d_1}(1-\kappa) >0$, concluding the proof.
\end{proof}

Let $(u_i)_{i\geq 0}$ denote the sequence  of words on $\Sigma$ in the length lexicographic order, so that $u_0=\varepsilon$, $u_1=a$, $u_2=b$, $u_3=aa$, \ldots 
For any $n$-state transition structure $\A=(n,\delta)$ that satisfies a template $\T=(n,\gamma)$ of $\tempB_n(p)$
and whose forward process is successful and halts, at step $i$, we associate the template $\T_\A=(n,\lambda)$ as follows. 
For any $x\in[n]$ and $\alpha\in\Sigma$,
\begin{itemize}
\item for all $j<i$ and all $\alpha \in \{a,b\}$ such that $u_j\alpha$ belongs $\{u_0,\ldots,u_i\}$,  $\lambda(x,\alpha)=\{\delta(x,\alpha)\}$ with $x= \delta(p,a u_j)$,
	\item otherwise, $\lambda(x,\alpha)= \gamma(x,\alpha)$.
\end{itemize}
In other words, starting from the template $\T$, we force the targets of all the transitions explored during the forward process. 
It is direct to establish that every transition structure that satisfies $\A$ has the same $\lceil\log_2\sqrt{n}\,\rceil$-backward substructure than $\A$ and exactly the same  forward process, which is successful.

Let $\tempC_n(p)$ denote the set of templates $\T_\A$ for the transition structures $\A$ that satisfy $\T\in\tempB_n(p)$ and that have a successful forward process from $p$. For such a transition structure $\A$ whose process halts after $i$ steps, define $\SupportC{\A}=\Support{\A}\cup\{\delta(p,au_j):0\leq j<i\}$ be the set of states visited during the backward construction and the forward process.

The next lemma is a direct consequence of Lemma~\ref{lm:associated template} and of the way we constrained the forward tree in the construction of $\T_\A$.

\begin{lem}\label{lm:associated cycle template}
Two $n$-state transition structures $\A=(n,\delta)$ and $\A'=(n,\delta')$ satisfy the same template $\T\in\tempC_n(p)$ if and only if 
$\A$ and $\A'$ have same $\lceil\log_2\sqrt{n}\,\rceil$-backward substructure and the forward processes from $p$ of both $\A$ and $\A'$ halt successfully at the same step $i$, with $\delta(p,au_j)=\delta'(p,au_j)$ for every $j\in\{0,\ldots, i-1\}$.
\end{lem}
For $\T_\A\in\tempC_n(p)$ a template  associated to a transition structure $\A=(n,\delta)$ whose  forward process from $p$  halts successfully after $i$ steps, we let $\SupportC{\T_\A}=\SupportC{\A}$ denote the set $\Support{\A}\cup\{\delta(p,au_j): 0\leq j < i\}$, which we call the \emph{cycle support of $\T_\A$} (or of $\A$). We added the states discovered during the process  to $\Support{\A}$.

Finally, observe that the last step of a successful process builds a cycle around $p$ as $\delta(p,au_i)\in L_h$ and as there is a path from any state of $L_h$ to $p$, by construction. Let $v$ be the smallest word for the length-lexicographic order that labels a path from $\delta(p,au_i)\in L_h$ to $p$, then $au_iv$ labels a cycle around $p$ in $\A$. Define $w_\A(p):=u_iv$, or just $w_\A$  if $p$ is clear from the context.

Moreover, this word only depends on $\T_\A$ since it only uses transitions determined by the template, by construction of the backward and forward traversals: two transition structures $\A$ and $\A'$ satisfying the same  $\T\in\tempC_n(p)$ are such that $w_\A=w_{\A'}$. So we can define this word for a given $\T\in\tempC_n(p)$ as $w_\T:=w_\A$ with no ambiguity.

The properties we need in the sequel are summarized in the next statement.

\begin{prop}\label{pro:cycle} There exists $c>0$ such that for $n$ sufficiently large and $p\in[n]$, a random $n$-state transition structure satisfies $\tempC_n(p)$ with probability at least $c$.

 Furthermore, any transition structure $\A$ that satisfies a template $\T=(n,\lambda)\in\tempC_n(p)$ admits a cycle around $p$ labeled by the word $aw_\T$, which only visits states in $\SupportC{\T}$ and has length between $h$ and $2h$, where $h=\lceil\log_2\sqrt{n}\,\rceil$.
We also have  $|\SupportC{\T}|\leq (d_1+1)\sqrt{n}$, and
for any $x\notin\SupportC{\T}$ and any $\alpha\in\Sigma$,  $\lambda(x,\alpha)=L_h(\T)\cup ([n]\setminus\Support{\T})$.  
\end{prop}

\begin{proof}
As both the $h$-backward substructure construction and the forward process are deterministic, a given $n$-state transition structure can 
satisfy at most one template in $\tempC_n(p)$. Hence, if we partition according to the elements of $\tempC_n(p)$, the law of total probabilities and Lemma~\ref{lm:forward process} yield the first statement of the proposition.

The second part is a consequence of the construction of $\T_\A$ and of the fact that if the forward process is successful, it halts before $\sqrt{n}$ states are discovered, and thus after building a forward tree of depth at most $h$.
\end{proof}

\section{Forming the \texorpdfstring{$b$-cycles}{b-cycles}}\label{sec:b-cycles}
From now on, we fix some  integer $d\geq 1$. In this section, we consider a random transition structure conditioned to satisfying a template $\T \in \tempC_n(p)$ with cycle around $p$ labeled by $aw$, with $w:=w_\T$ and $|w|\in\Theta(\log n)$. The statements of this part articulate as follows.

We consider the  path $\mathcal{P}=q\stackrel{w}{\rightsquigarrow}p_1\stackrel{aw}{\rightsquigarrow}p_2\rightsquigarrow\cdots\stackrel{aw}{\rightsquigarrow} p_d$ depicted in Figure~\ref{fig:proof sketch}. We first establish in Lemma~\ref{lm:starting states} that with high probability it does not intersect the support of the template $\T$ and it does not go twice through the same state.

Starting at state $x$, we grow what we call a $b$-thread by successively drawing the outgoing $b$-transitions until we cycle back.
Then in Lemma~\ref{lm:thread}, we show that with visible probability we can grow a $b$-thread from $p$ and each $p_i$ such that these threads have size in $O(\sqrt{n})$, do not intersect the support of the template $\T$ nor the path $\mathcal{P}$ and are pairwise disjoint.

If we condition the cycles formed by the $b$-threads to have length in $\intinter{\sqrt{n}}{2\sqrt{n}\,}$, we furthermore establish that these lengths are uniformly drawn in this interval and that it still happens with visible probability in Proposition~\ref{pro:loops}.

\subsection{Growing the \texorpdfstring{$b$-threads}{b-threads}}\label{sec:b-threads}

Let $\A=(n,\delta)$ be a $n$-state transition structure that satisfies a template $\T\in\tempC_n(p)$, for some $p\in[n]$. For any $q\in[n]$, we associate to $\A$ its \emph{starting states tuple} $s(\A,q)$ which is either a tuple of $d+1$ states, or $\bot$ in case the construction of the tuple failed. To build $s(\A,q)$ consider the path labeled in $\A$ that starts at state $q$ and labeled by the word \[
u=w_\A(aw_\A)^{d-1}\]
where $aw_\A$ labels the cycle around $p$ in $\A$ defined in the previous sections. The word $u$ has length $|u|=d|w_\A|+d-1$, and go through all the states $x_i:=\delta(q,v_i)$, where $v_i$ is the prefix of length $i$ of $u$ and $0\leq i \leq |u|$. {Letting} $X_q(\A)=\{x_i:0\leq i \leq |u|\}$, the construction fails when $X_q(\A)\cap\SupportC{\A}\neq \emptyset$ or $|X_q(\A)|<|u|+1$; in other words, it fails if the path uses a state of $\SupportC{\A}$ or if it goes twice through the same state. If the construction fails, we set $s(\A,q)=\bot$, otherwise 
\[
s(\A,q):= \left(p,\delta(q,w_\A), \delta\left(q,w_\A(aw_\A)^1\right), \ldots, \delta\left(q,w_\A(aw_\A)^{d-1}\right)\right).
\]

\begin{lem}\label{lm:starting states}
For any $\epsilon>0$, for any $n$ sufficiently large, for all $p\in [n]$, if $q$ is chosen uniformly at random in $[n]$, then for any $\T\in\tempC_n(p)$ we have
\[
\proba\left(s(\A,q)\neq\bot\mid \A\satisfies\T\right) \geq 1-\epsilon.
\]
\end{lem}

\begin{proof}
Observe that the construction immediately fails for $\A\satisfies\T$ if the starting state of the path $x_0=q$ is in $\SupportC{\T}$, as $\SupportC{\A}=\SupportC{\T}$. This happens with probability $|\SupportC{\T}|/n$ as $q$ is chosen  uniformly at random in $[n]$. As stated in Proposition~\ref{pro:cycle}, in  a random $\A$ satisfying $\T$, the transitions starting from $x \not\in \SupportC{\T}$ must arrive in $S:=L_h(\T)\cup([n]\setminus\Support{\T})$. Furthermore their targets are all chosen uniformly at random and independently in $S$, according to Lemma~\ref{lm:proba template}. For $i\in\{0,\ldots,|u|\}$, let $Y_i:=\SupportC{\T}\cup\{x_0,\ldots,x_{i}\}$.
Observe that for any $i\in\{1,\ldots,|u|\}$, if the path for the first $i-1$ letters of $u$ did not produce a failure, the target of a transition starting from $x_{i-1}$ is an element of $S$ taken uniformly at random. Hence, we have, as $|\SupportC{\T}\cap S|\leq |\SupportC{\T}|$:
\[
\proba(x_i\notin Y_{i-1}\mid \forall j\leq i-1,
x_{j}\notin Y_{j-1}) \geq 1 - \frac{|\SupportC{\T}|+i}{|S|}
\geq 1-\frac{|\SupportC{\T}|+|u|}{|S|}.
\]
Combining with the probability that $x_0=q\notin\SupportC{\T}$ this yields
\[
\proba(s(\A,q)\neq \bot\mid \A\satisfies \T) \geq \left(1-\frac{|\SupportC{\T}|+|u|}{|S|}\right)^{|u|+1}.
\]
By Proposition~\ref{pro:cycle}, $|u|\leq 2\lceil\log_2\sqrt{n}\,\rceil$ and $|\Support{\T}|\leq|\SupportC{\T}|\leq (d_1+1)\sqrt{n}$, so that $|S|\geq n-(d_1+1)\sqrt{n}$ and we have
\[
\proba(s(\A,q)\neq \bot\mid \A\satisfies \T) \geq 
\left(1-\frac{(d_1+1)\sqrt{n}+2\lceil\log_2\sqrt{n}\,\rceil}{n-(d_1+1)\sqrt{n}}\right)^{2\lceil\log_2\sqrt{n}\,\rceil+1}.
\]
This concludes the proof, as it tends to $1$ when $n$ tends to infinity.
\end{proof}

If $\A$ satisfies $\T$ and $s(\A,q)\neq\bot$, we write $s(\A,q)=(p_0,p_1,\ldots,p_{d})$, with $p_0=p$ by construction. 
{Let $u_i$ denote the $i$-th letter of the word $u$ and} let also $\mathcal{P}_q(\A)=x_0\xrightarrow{u_1} x_1 \xrightarrow{u_2} \cdots \xrightarrow{u_{|u|}} x_{|u|}$ denote the path labeled by $u=u_1u_2\cdots u_{|u|}$ and starting from $x_0=q$. Recall that $X_q(\A)$ is the set of states of $\mathcal{P}_q(\A)$.

We consider the \emph{thread process} of $\A$ which, if successful, consists in building in order the sets
$E_i=\{\delta(p_i,b^j):j\geq 0\}$ for $i\in\{0,\ldots,d\}$ as follows.
For $i$ from $0$ to $d$ we start with $E_i=\{p_i\}$ and then iteratively
add $\delta(p_i,b^j)$ for $j\geq 0$ until:
\begin{enumerate}
	\item $\delta(p_i,b^j)\in \cup_{k=0}^{i-1} E_k$, in which case the process halts and is a failure;\label{thread not disjoint}
	\item or $\delta(p_i,b^j)\in \SupportC{\T}\cup X_q(\A)$, in which case the process halts and is a failure;\label{thread already seen}
	\item or $\delta(p_i,b^j)\in E_i$, in which case the process halts and is a failure if $|E_i|\notin \intinter{2\sqrt{n}}{3\sqrt{n}\,}$. If not there are two cases: if {$i<d$} the process starts building $E_{i+1}$, and if {$i=d$} it halts with a success, as all the $b$-threads are successfully built.
\end{enumerate}
Define $\lambda(\A,q)=(|E_0|, |E_1|, \ldots, |E_{d}|)$  the tuple of the $b$-threads' lengths  
if the process is successful and $\lambda(\A,q)=\bot$ if it fails. Observe that if successful, then the $E_i$'s are pairwise disjoint, and do not intersect $\SupportC{\T} \cup X_q(\A)$.

 The following statement is a variation on the classical Birthday Problem.
 
 \begin{lem}\label{lm:thread}
 There exists $c_\lambda>0$ such that for any
 $n$ sufficiently large, for any $p\in[n]$, for any $\T\in\tempC_n(p)$  and for any path $\P=x_0\xrightarrow{u_1}x_1\xrightarrow{u_2}\ldots\xrightarrow{u_{|u|}}x_{|u|}$ labeled by $u=w_\T(aw_\T)^{d-1}$ such that the $x_i$'s are pairwise distinct and not in $\SupportC{\T}$ we have
\[
\proba\left(\lambda(\A,x_0)\neq \bot \mid \A\satisfies \T \text{ and }\P_{x_0}(\A)=\P\right) \geq c_\lambda.
\]
\end{lem}

\begin{proof}
Conditioning by $\A\satisfies\T$ and $\mathcal{P}_{x_0}(\A)=\mathcal{P}$ is exactly conditioning on satisfying the template $\T_\mathcal{P}=(n,\gamma_\mathcal{P})$ obtained from $\T=(n,\gamma)$ by setting
$\gamma_\mathcal{P}(x_i,u_{i+1})=\{x_{i+1}\} $ for all $i\in\{0,\ldots, |u|-1\}$, and $\gamma_\mathcal{P}(x,\alpha)= \gamma(x,\alpha)$ for all other transitions.

Let $S:=L_h(\T)\cup([n]\setminus\SupportC{\T})$.
For every $i\in\{1,\ldots,d\}$, by construction $p_i=x_{i|aw_\T|-1}$ and $\gamma_\mathcal{P}(p_i,b) = \gamma(p_i,b) = S$ since $a$ labels the transition outgoing from  $p_i$ in $\mathcal{P}$ for $i<d-1$ and $p_{d}$ is the end of $\mathcal{P}$. Also,  $\gamma_\mathcal{P}(p_0,b)= \gamma(p_0,b)=\gamma(p,b)=[n]$, as the template gives no constraint on $p\xrightarrow{b}$.

So during the thread process, as long as we are adding a new state $x=\delta(p_i,b^j)$ in $E_i$, the transition $\delta(x,b)$ is a uniform element of $S$ (or of $[n]$ if $i=j=0$), independently of the previous steps. 
Hence if we are not in the case $i=j=0$, out of the $|S|$ possibilities, $r_i:=|\cup_{k=0}^{i-1}E_k|$ possibilities trigger a failure because of Condition~\eqref{thread not disjoint}, $s:=|S\,\cap\,\SupportC{\T}|+|X_q(\A)|$ possibilities trigger a failure because of Condition~\eqref{thread already seen}, and $j$ possibilities complete the process for $E_i$ because we cycle back on a previously seen element of $E_i$. So for $i\geq 1$ and $t <3\sqrt{n}$, the probability $\pi_i(t)$ that  neither Condition~\eqref{thread not disjoint} nor Condition~\eqref{thread already seen} were triggered and $|E_i|= t$, conditioned on the fact that  $E_j$  were successfully built for all $j<i$ is 
\[
\pi_i(t)=\underbrace{\frac{|S|-r_i-s}{|S|}}_{\text{choice of }\delta(p_i,b)}\times \underbrace{\frac{|S|-r_i-s-1}{|S|}}_{\text{choice of }\delta(p_i,b^2)}\times \cdots
\times \underbrace{\frac{|S|-r_i-s-(t-2)}{|S|}}_{\text{choice of }\delta(p_i,b^{t-1})}\times \underbrace{\frac{t}{|S|}}_{\text{cycling back}}.
\]
Hence
\[
\pi_i(t) = \frac{t}{|S|}\,\prod_{j=0}^{t-2}\left(1-\frac{r_i+s+j}{|S|}\right).
\]
For $2\sqrt{n}\leq t \leq 3\sqrt{n}$ we have
\[
\pi_i(t) \geq \frac{t}{|S|}\,\left(1-\frac{r_i+s+3\sqrt{n}}{|S|}\right)^{3\sqrt{n}}.
\]
Moreover, as at most $d$ threads have been completed so far, $r_i\leq 3d\sqrt{n}$, and for $n$ sufficiently large, $|X_q(\A)|\leq\sqrt{n}$ and thus $s\leq (d_1+2)\sqrt{n}$, by Proposition~\ref{pro:backward}.  Observe that we also have $|S|\geq n-(d_1+1)\sqrt{n}$, which is greater than $\frac12n$ for $n$ sufficiently large. 
Therefore, using $\mu:=2d_1+6d+10$ and the fact that $\lim_m (1-\mu/m)^m=\exp(-\mu)$, we have
\[
\pi_i(t) \geq \frac{t}{|S|}\,\left(1-\frac{\mu}{\sqrt{n}}\right)^{3\sqrt{n}}\geq
\frac{te^{-3\mu}}{2n}. 
\]
This yields that
\[
\sum_{t=\lceil2\sqrt{n}\rceil}^{\lfloor 3\sqrt{n}\rfloor}\pi_i(t)
\geq \frac{e^{-3\mu}}{2n}\sum_{t=\lceil2\sqrt{n}\rceil}^{\lfloor 3\sqrt{n}\rfloor}t \geq  \frac{e^{-3\mu}}{2n} \left(\lfloor 3\sqrt{n}\rfloor - \lceil2\sqrt{n}\,\rceil+1\right)\lceil2\sqrt{n}\rceil.
\]
As the limit of the right hand term is $e^{-3\mu}>0$, there exists $\kappa>0$ such that, for $n$ sufficiently large we have 
$
\sum_{t=\lceil2\sqrt{n}\rceil}^{\lfloor 3\sqrt{n}\rfloor}\pi_i(t)\geq \kappa.
$

For the case $i=0$, as $r_0=0$ and $\gamma_\P(p_0,b)=[n]$, we have almost the same formula:
\[
\pi_0(t)= \frac{t}{|S|}\left(1-\frac{s}{n}\right)\,\prod_{j=1}^{t-2}\left(1-\frac{s+j}{|S|}\right).
\]
With the same technique as before, we can find some positive constant $\kappa_0$ such that $\sum_{t=\lceil2\sqrt{n}\rceil}^{\lfloor 3\sqrt{n}\rfloor}\pi_0(t)\geq \kappa_0$, for $n$ sufficiently large. Combining the results we obtain that
\[
\proba\left(\lambda_{x_0}(\A)\neq \bot \mid \A\satisfies \T \text{ and }\P_{x_0}(\A)=\P\right) \geq \kappa_0\,\kappa^d,
\]
concluding the proof.
\end{proof}

We now write the conditions in terms of templates, by naming the states encountered during the thread process and forcing the associated transitions. For $n\geq 1$ and $p,q\in[n]$, let $\tempT_n(p,q)$ denote the set of templates $\T=(n,\gamma)$ such that:
\begin{itemize}
    \item $\T$ satisfies a template $\hat\T=(n,\tilde\gamma)$ of $\tempC_n(p)$, with a cycle around $p$ labeled by the word $aw_\T$. Hence $\SupportC{\T}=\SupportC{\hat\T}$ and, to simplify the notations, define $w:=w_\T=w_{\hat\T}$.
    \item We write $x_0=q$ and for $u=w(aw)^{d-1}=u_1u_2\cdots u_{|u|}$, there exist pairwise distinct states $x_1$, $x_2$, \ldots, $x_{|u|}$ in $[n]\setminus\SupportC{\T}$ such that for all $i\in\{0,\ldots,|u|-1\}$,   $\gamma(x_i,u_{i+1})=\{x_{i+1}\}$. This determines the path starting at $q$ and labeled by $u$ of any $\A$ that satisfies $\T$. Let $X_q(\T)=\{x_i:0\leq i\leq |u|\}$. 
    \item We write $p_{0,1}=p$ and $p_{i,1}=x_{i|aw|-1}$ for $1\leq i\leq d$. For any $i\in\{0,\ldots, d\}$, there exists an integer $\lambda_i$ such that $2\sqrt{n}\leq \lambda_i\leq 3\sqrt{n}$ and states $p_{i,2}$, \ldots, $p_{i,\lambda_i}$ in $[n]\setminus(\SupportC{\T}\cup X_q(\T))$ such that:
    \begin{itemize}
        \item The $p_{i,j}$'s are pairwise distinct.
        \item For all $i\in\{0,\ldots,d\}$, for all $j\in\{1,\ldots,\lambda_j-1\}$, $\gamma(p_{i,j},b)=\{p_{i,j+1}\}.$ This determines the $b$-threads starting from $p_{0,1}$, \ldots, $p_{d,1}$. Let $E_i=\{p_{i,j}:j\in[\lambda_i]\}$.
        \item For all $i\in\{0,\ldots,d\}$, $\gamma(p_{i,\lambda_i},b)=E_i$, to ensure we cycle back in each $b$-thread.
    \end{itemize}  
    \item For every other transition $(x,\alpha)$ we have $\gamma(x,\alpha)=\hat\gamma(x,\alpha)$.
\end{itemize} 
As the construction follows the constructions of the path and of the $b$-threads, we readily have that $\A\satisfies\tempT_{n}(p,q)$ 
if and only if $\A\satisfies\tempC_p(n)$ and $\lambda(\A,q)\neq\bot$. 
As the constructions are deterministic, for given $p,q\in [n]$, a given $n$-state transition structure  $\A$ cannot satisfy more than one template of $\tempT_{n}(p,q)$. 
Moreover, all the transition structures $\A$ that satisfy a given $\T\in\tempT_{n}(p,q)$ have same $\lambda(\A,q)=(\lambda_0,\lambda_1,\ldots,\lambda_d)$, which we can therefore write $\lambda(\T,q)$ with no ambiguity.

\begin{prop}\label{pro:threads}
There exists  $c_\textrm{T}>0$ such that, for $n$ sufficiently large and for $p\in[n]$, if $q$ is a uniform random element of $[n]$ and $\A$ is a uniform $n$-state transition structure, taken independently, then 
$\proba\left(\A\satisfies\tempT_{n}(p,q)\right) \geq c_\textrm{T}.$
 \end{prop}

 \begin{proof}
We partition  $\tempT_{n}(p,q)$ according to the possible valid paths as follows. For given $\T\in\tempT_{n}(p,q)$, the path $\P_q(\A)$ is the same for every $\A$ that satisfies $\T$, so we can define $\P_q(\T):=\P_q(\A)$ for any such $\A$.
Let $\mathfrak{P}_q(n)$ denote the set of all possible paths:
\[
\mathfrak{P}_q(n) = \{\P_q(\T):\T\in\tempT_{n}(p,q)\}.
\]
We have
\[
\proba\left(\A\satisfies \tempT_{n}(p,q)\right)
= \sum_{\P\in\mathfrak{P}_q(n)}
\sum_{\T\in\tempT_{n}(p,q)} \proba\left(\A\satisfies \T \text{ and }\P_q(\T)=\P\right).
\]
 Recall that $\A\satisfies\tempT_{n}(p,q)$ if and only if $\A\satisfies \tempC_{p}(n)$ and
$\lambda(\A,q)\neq\bot$. This yields
\[
\proba\left(\A\satisfies \tempT_{n}(p,q)\right)
= \sum_{\P\in\mathfrak{P}_q(n)}
\sum_{\T\in\tempC_p(n)} \proba\left(\A\satisfies \T \text{ and } \lambda(\A,q)\neq\bot\text{ and }\P_q(\T)=\P\right).
\]
For any $\T\in\tempC_p(n)$ and any $\P\in\mathfrak{P}_q(n)$, by Lemma~\ref{lm:thread}, for $n$ sufficiently large we have
\[
\proba\left(\A\satisfies \T \text{ and }  \lambda(\A,q) \neq\bot\text{ and }\P_q(\T)=\P\right)
\geq c_{\lambda} \, \proba(\A\satisfies \T \text{ and }\P_q(\T)=\P).
\]
By definition of $s_q(\A)$, a given transition structure that satisfies $\T\in\tempC_p(n)$ has a path in $\mathfrak P_q(n)$ if and only if $s_q(\A)\neq \bot$. Thus
\[
\sum_{\P\in\mathfrak{P}_q(n)}\proba(\A\satisfies \T \text{ and }\P_q(\T)=\P) =
\proba(\A\satisfies \T \text{ and }s_q(\A)\neq \bot).
\]
Moreover, taking $\epsilon=\frac12$ in Lemma~\ref{lm:starting states}, for $n$ sufficiently large we have
\[
\proba(\A\satisfies \T \text{ and }s_q(\A)\neq \bot) \geq \frac12 \ \proba(\A\satisfies \T).
\]
Putting all together yields
\[
\proba\left(\A\satisfies \tempT_{n}(p,q)\right)
\geq \frac{c_\lambda}{2} \sum_{\T\in\tempC_p(n)} \proba(\A\satisfies \T)
= \frac{c_\lambda}{2} \proba(\A\satisfies \tempC_p(n)).
\]
This concludes the proof by Proposition~\ref{pro:cycle}.     
 \end{proof}

Let $\T\in\tempT_{n}(p,q)$ with $\lambda(\T,q)=(\lambda_0,\ldots,\lambda_d)$. If $\A\satisfies\T$ then at the end of each $b$-thread from $p_{i,1}$, the $b$-transition outgoing from $p_{i,\lambda_i}$ ends in an element $p_{i,j}$ of $E_i$, forming a $b$-cycle of length $\ell_i=\lambda_i-j+1$. 
Let $\ell(\A,p,q)$ be the $d+1$-tuple of the $b$-cycle lengths $(\ell_0,\ldots,\ell_d)$ if all its coordinates are in $\intinter{\sqrt{n}}{2\sqrt{n}\,}$, and $\ell(\A,p,q)=\bot$ otherwise. 
We also set $\ell(\A,p,q)=\bot$ if $\A$ does not satisfy $\tempT_{n}(p,q)$. If $\ell(\A,p,q)=(\ell_0,\ldots,\ell_d)$, we define $\ell_i(\A):=\ell_i$ to directly access its coordinates.

\begin{prop}\label{pro:loops}
There exists $c_\ell>0$ such that for $n$ sufficiently large and $p\in[n]$, if $q$ is taken uniformly in $[n]$ and $\A$ is a uniform $n$-state transition structure taken independently then $\proba\left(\ell(\A,p,q)\neq\bot\right)\geq c_\ell$.
Moreover, conditioned on $\ell(\A,p,q)\neq\bot$, the random variables $\ell_i$ are independent uniform integers of $\intinter{\sqrt{n}}{2\sqrt{n}\,}$.
\end{prop}

\begin{proof}
By definition, we have $\ell(\A,p,q)\neq\bot$ only for transition structures $\A$ such that 
$\A\satisfies\tempT_{n}(p,q)$. Let $\T\in\tempT_{n}(p,q)$, let $s_q(\T)=(p_0,p_1,\ldots,p_d)$ denote its tuple of starting states and let
$\lambda_q(\T)=(\lambda_0,\ldots,\lambda_p)$ denote its tuple of $b$-thread lengths. If $\A=(n,\delta)$ is a random transition structure conditioned on {satisfying} $\T$, then, by definition of $\tempT_{n}(p,q)$, for every $i\in\{0,\ldots,d\}$, 
$\delta(p_i,b^{\lambda_i})$ 
is a uniform element of $E_i=\{\delta(p_i,b^j):j\geq 0\}$, with $|E_i|=\lambda_i$, and they are all independent. Hence each $b$-cycle length $\hat\ell_i$ is a uniform random integer in $[\lambda_i]$, and they are all independent. 

By definition of $\tempT_{n}(p,q)$, for every $i\in\{0,\ldots,d\}$, we have $2\sqrt{n}\leq \lambda_i \leq 3\sqrt{n}$. In particular $\intinter{\sqrt{n}}{2\sqrt{n}}\subseteq\intinter{1}{\lambda_i}$, and the probability that a uniform element of $\intinter{1}{\lambda_i}$ is in $\intinter{\sqrt{n}}{2\sqrt{n}}$ is at least $\frac14$ (it is lower-bounded by a quantity that tends to $\frac13$). Hence
the probability that all the $\hat\ell_i$'s are in the valid range $\intinter{\sqrt{n}}{2\sqrt{n}}$ is at least $4^{-d-1}$. Since this lower bound does not depend on $\T\in\tempT_{n}(p,q)$, this proves by the law of total probabilities that 
\[
\proba\left(\ell(\A,p,q)\neq\bot\right)
= \!\!\sum_{\T\in\tempT_{n}(p,q)} \proba\left(\ell(\A,p,q)\neq\bot\mid \A\satisfies\T\right)\proba(\A\satisfies\T)
\geq 4^{-d-1} \proba(\A\satisfies\tempT_{n}(p,q)).
\]
By Proposition~\ref{pro:threads}, this yields $\proba\left(\ell(\A,p,q)\neq\bot\right)\geq c_\ell$, with $c_\ell=4^{-d-1}c_\mathrm{T}$.

The second part of the statement is just a consequence of the fact that a uniform element of $[\lambda_i]$ conditioned to be in $\intinter{\sqrt{n}}{2\sqrt{n}}$ is a uniform element of $\intinter{\sqrt{n}}{2\sqrt{n}}$. And as direct computation shows, the independence is preserved if we consider all the $\ell_i$'s together.
\end{proof}

Proposition~\ref{pro:loops} is our main probabilistic result on random transition structures. It states that our global construction {succeeds} with visible probability and produces $d+1$ $b$-cycles of uniform and independent length in $\intinter{\sqrt{n}}{2\sqrt{n}\,}$. These cycles are linked to $p$ in a way that can be exploited during the accessible subset construction when we add the transition $p\xrightarrow{a}q$.

\section{Super-polynomial growth of the subset construction}\label{sec:determinisation}

A uniform $n$-state almost deterministic transition structure $\A=(n,\delta,p\xrightarrow{a}q)$ is obtained by choosing $\delta$, $p$ and $q$ uniformly at random and independently. If we furthermore choose the starting state $i_0$ uniformly at random and independently, we can use Proposition~\ref{pro:loops} to establish that the accessible subset construction has a super-polynomial number of states with visible probability. This is the focus of this section.

First, Grusho~\cite{grusho73} established that with high probability a uniform $n$-state transition structure has a unique terminal strongly connected component. When there is only one such strongly connected component, it is necessarily accessible from any state. We also rely on a result of Cai and Devroye~\cite[Theorem~2]{devroye17}, which ensures that with high probability there are no cycles of length {larger} than $\frac14\log_2n$ outside this unique terminal strongly connected component. 
More precisely, they show that the length of the longest cycle outside  the unique terminal strongly connected component of a random $n$-state transition structure is in $O_p(1)$ with the notations of \cite{randomgraphs}. By \cite[Remark 1.3]{randomgraphs}, this in particular implies that for any $\omega(n) \rightarrow \infty$, all the cycles outside the accessible part have length at most $\omega(n)$ with high probability.

Observe that if $\A\satisfies\tempT_{n}(p,q)$, then it has a  word $aw_\A$ that labels a cycle around $p$ of length at most $h=\lceil \log_2\sqrt{n}\,\rceil$. By Proposition~\ref{pro:loops}, this happens with visible probability. The  conjunction of a high-probability event with a visible event being a visible event, we  directly have the following lemma.

\begin{lem}\label{lm:accessible p}
If $\A$ is a  $n$-state transition structure, and $p,q,i_0$ are  states in $[n]$, all being taken uniformly at random and independently in their respective sets, then with visible probability, $\ell(\A,p,q)\neq\bot$ and $p$ is accessible from $i_0$. 
\end{lem}

At this point, to prove our first theorem, on the accessible powerset  construction applied to a random transition structure with an added transition, we need a result of probabilistic number theory. 
Tóth in~\cite{toth2002}  generalizes  the folklore result 
that two independent random numbers in $[N]$ are coprime with probability that tends to $\frac{6}{\pi^2}$ to a fixed number of independent random numbers.
\begin{thmC}[\cite{toth2002}]\label{th:toth}
For any $k\geq 2$, there exists some  constant $A_k>0$ such that $k$ integers taken uniformly at random and independently in $[N]$ are pairwise coprime with probability that tends to $A_k$ as $N$ tends to infinity, with
\[
A_k=\prod_{p \;\textrm{prime}} \left(1+\frac{k}{p-1}\right)\left(1-\frac{1}{p}\right)^k.
\]
\end{thmC}

This theorem was proven using probabilistic arguments in \cite[Theorem 3.3]{CaiB03}. We adapt this proof to obtain the following corollary.

\begin{cor}\label{cor:pairwise coprime}
For any $d\geq 1$,  $d+1$ integers taken uniformly at random and independently in $\intinter{\sqrt{n}}{2\sqrt{n}}$ are pairwise coprime with probability that tends to $A_{d+1}$ as $n$ tends to infinity.
\end{cor}

\begin{proof}
For any real $x>0$, let $P_x$ be the set of prime numbers smaller than or equal to $x$: $P_x=\{p\in\mathbb{Z}_{\geq 1}: p\leq x\text{ and }p\text{ prime}\}$.
For any $n\geq 2$, let $i_n$ and $x_n$ be  defined by
\[
i_n = \max\{i\in\mathbb{Z}_{\geq 0}: \prod_{p\in P_i}p\leq \sqrt[3]{n}\}\quad\text{ and }\quad x_n=\prod_{p\in P_{i_n}}p.
\]
Observe that $(i_n)_{n\geq2}$ and $(x_n)_{n\geq 2}$ both 
tend to infinity as $n$ tends to infinity. 

We now follow the idea of Cai and Bach~\cite{CaiB03}. We consider the uniform distribution on $\intinter{\sqrt{n}}{2\sqrt{n}\,}^{d+1}$ and define the following events:
\begin{itemize}
\item $E_n$: the $d+1$ integers are pairwise coprime,
\item for all prime $p$, $G^p_n$: at most one of the $d+1$ integers is divisible by $p$,
\item for all prime $p$, $B^p_n$: at least two of the $d+1$ integers are divisible by $p$.
\end{itemize}
Notice that $G^p_n$ is the complement of $B^p_n$, and that $E_n=\bigcap_{\substack{p\leq 2\sqrt{n}\\p\text{ prime}}}G^p_n$. By separating the cases $p\leq i_n$ and $p>i_n$, let us introduce the events $R_n=\bigcap_{\substack{p\leq i_n\\p\text{ prime}}}G^p_n$ and $S_n=\bigcap_{\substack{p> i_n\\p\text{ prime}}}G^p_n$, so that $E_n=R_n\cap S_n$, and $\overline{S_n}=\bigcup_{\substack{p> i_n\\p\text{ prime}}}B^p_n$. Immediately, $\proba(E_n)\leq \proba(R_n)$, and as $\proba(R_n\cap S_n)\geq \proba(R_n)-\proba(\overline{S_n})$, we have, using the union bound:
\begin{equation}\label{eq:Cai}
\proba\bigg(\bigcap_{\substack{p\leq i_n\\p\text{ prime}}}G^p_n\bigg) - \sum_{\substack{p >i_n\\p\text{ prime}}} \proba\left(B^p_n\right) 
\leq \proba\left(E_n\right)\leq \proba\bigg(\bigcap_{\substack{p\leq i_n\\ p \text{ prime}}}G^p_n\bigg).
\end{equation}
{The end of the proof consists in two steps. First, we show that $\sum_{\substack{p >i_n\\p\text{ prime}}} \proba\left(B^p_n\right)$ tends to $0
$. Then we show that $\proba\Big(\bigcap_{\substack{p\leq i_n\\p\text{ prime}}}G^p_n\Big)$ tends to $A_{d+1}$.}
In an integer interval of length $m$ there are at most $m/p+1$ multiples of $p\geq 2$. Hence, by the union bound, if $m_n=\lfloor{2\sqrt{n}\rfloor-\lceil\sqrt{n}\rceil+1}$ denote the length of $\intinter{\sqrt{n}}{2\sqrt{n}}$:
\[
\sum_{\substack{p>i_n\\p\text{ prime}}} \proba\left(B^p_n\right) \leq \binom{d+1}{2} \sum_{\substack{p=i_n+1\\p\text{ prime}}}^{\lfloor2\sqrt{n}\rfloor}\left(\frac1p+\frac1{m_n}\right)^2
\leq \binom{d+1}{2}\sum_{p\geq i_n+1} \frac{9}{p^2},
\]
as $2m_n\geq p$. The sum  $\sum_{p\geq i_n+1}p^{-2}$ is the remainder of order $i_n\to\infty$ of a convergent series, hence
\[
\lim_{n\rightarrow\infty}\sum_{p>i_n} \proba\left(B^p_n\right) = 0.
\]

For $n\geq1$, let  $I_n:=\intinter{a_n}{b_n}$ be the largest interval of the form
$\intinter{ix_n+1}{jx_n}$ that is included in $\intinter{\sqrt{n}}{2\sqrt{n}}$:
\begin{align*}
a_n &= \min\{i x_n+1: i\in\mathbb{Z}_{\geq 0}\text{ and }i x_n+1\geq\sqrt{n}\},\\
b_n &= \max\{j x_n: j\in\mathbb{Z}_{\geq 0}\text{ and }j x_n\leq2\sqrt{n}\}.
\end{align*}
Observe that $I_n\neq\emptyset$ for $n$ sufficiently large, as  after some point   $\sqrt[3]{n}$ is less than $\sqrt{n}/3$.

For any non-empty integer interval $I$, we denote by $X_{I}$ the random variable  uniformly distributed on $I$. Observe that for any set of prime numbers $\{p_1,\ldots,p_t\}\in P_{i_n}$, 
the number of multiples in $I_n$ of the product $p_1\ldots p_t$ is exactly $\frac{{b_n-a_n+1}}{p_1\cdots p_t}$. Then we have, by direct counting
\[
\proba\left(\forall i\in[t],\ p_i\text{ divides } X_{I_n}\right) = \proba\left(p_1\cdots p_t\text{ divides } X_{I_n}\right)=\frac{1}{p_1\ldots p_t}=
 \prod_{i=1}^t \proba\left(p_i\text{ divides } X_{I_n}\right).
\]
Hence we built $I_n$ so that there is no rounding effects, to avoid slight alterations of the probabilities and get the independency. Moreover, for given $p\leq i_n$, the probability that none of the coordinates of a uniform random element of $I_n$ is divisible by $p$ is $(1-1/p)^{d+1}$, and that exactly one of them is divisible by $p$ is $(d+1)(1-1/p)^{d}/p=(1-1/p)^{d+1}(d+1)/(p-1)$. Hence, if we condition the integer vector to be in $I_n^{d+1}$ we have:
\[
\proba\bigg(\bigcap_{\substack{p\leq i_n\\p\text{ prime}}}G^p_n \mid I_n^{d+1}\bigg) = \prod_{\substack{p\leq i_n\\p\text{ prime}}}\left(1-\frac1p\right)^{d+1}\left(1+\frac{d+1}{p-1}\right).
\]
As $\left(1-\frac1p\right)^{d+1}\left(1+\frac{d+1}{p-1}\right)=1+O(p^{-2})$, the product has a limit as $n$, hence $i_n$, tends to infinity, which we call $A_{d+1}$  as in~\cite{toth2002}. Finally, observe that since we chose $x_n\leq\sqrt[3]{n}$, for a uniform vector of $\intinter{\sqrt{n}}{2\sqrt{n}}^{d+1}$ we have
\[
\proba\left(I_n^{d+1}\right) \geq  \left(\frac{m_n-2x_n}{m_n}\right)^{d+1}
\geq  \left(1-\frac2{n^{1/6}}\right)^{d+1}.
\]
Therefore, $\proba\left(I_n^{d+1}\right)=1-O(n^{-1/6})$. 
{
By the law of total probabilities we have
\[
\proba\bigg(\bigcap_{\substack{p\leq i_n\\p\text{ prime}}}G^p_n\bigg)=
\proba\bigg(\bigcap_{\substack{p\leq i_n\\p\text{ prime}}}G^p_n \mid I_n^{d+1}\bigg)  \proba\left(I_n^{d+1}\right) 
+ \proba\bigg(\bigcap_{\substack{p\leq i_n\\p\text{ prime}}}G^p_n \cap \overline{I_n^{d+1}}\bigg). 
\]
This concludes the proof by Equation~\eqref{eq:Cai} since
\[
\proba\bigg(\bigcap_{\substack{p\leq i_n\\p\text{ prime}}}G^p_n \mid I_n^{d+1}\bigg)  \proba\left(I_n^{d+1}\right) 
\xrightarrow[n\rightarrow\infty]{} A_{d+1},
\]
and
\[
\proba\bigg(\bigcap_{\substack{p\leq i_n\\p\text{ prime}}}G^p_n \cap \overline{I_n^{d+1}}\bigg)\leq
\proba\left( \overline{I_n^{d+1}}\right) = O(n^{-1/6}) \xrightarrow[n\rightarrow\infty]{} 0.
\]
}
This concludes the proof.
\end{proof}

If $\A=(n,\delta,p\xrightarrow{a}q)$ is an almost deterministic transition structure, let $\D(\A)$ denote its powerset transition structure, and let $\D(\A,i_0)$ denote its accessible powerset transition structure from the unique initial state $i_0\in[n]$.
Observe that if $\ell(\A,p,q)\neq\bot$, then $\A\satisfies\tempC_{p}(n)$, and the word $aw_\A$  labels a cycle around $p$. 
Hence, using the word $au=(aw_\A)^{d}$, as in Section~\ref{sec:forward}, we have $\delta(\{p\},au)=\{p_0,p_1,\ldots,p_d\}$, where $(p_0,\ldots,p_d)=s(\A,q)$ is the starting states tuple of $\A$. Hence, if $\A$ satisfies a template in $\tempT_{p,q}(n)$,
 for every $j\geq 0$, we have $\delta(\{p_0,p_1,\ldots,p_d\},b^j)=\{\delta(p_0,b^j),\ldots,\delta(p_d,b^j)\}$. %
For $j$ large enough, all the $\delta(p_i,b^j)$'s are in the cyclic part of their respective $b$-thread, forming a $b$-cycle of length $\textrm{lcm}(\ell_0(\A), \ldots, \ell_d(\A))$ in $\D(\A)$. 
Let $\C_b(\A)$ denote this $b$-cycle, which is defined when $\ell(\A,p,q)\neq\bot$.

\begin{thm}\label{thm:determinisation}
Let $\nu$ be a positive integer.
For the uniform distribution of almost deterministic transition structures with $n$-states with an added  initial state taken uniformly at random and independently,  the accessible powerset transition structure has more than $n^\nu$ states with visible probability. 

As a consequence, an algorithm that builds the accessible powerset automaton of a random almost deterministic automaton has super-polynomial average time complexity.
\end{thm}

\begin{proof}
Let $\A=(n,\delta,p\xrightarrow{a}q)$ be an almost deterministic transition structure such that $\ell(\A,p,q)\neq\bot$, and $p$ is accessible from $i_0\in[n]$, which acts as its initial state. 
Then there exists a word $v\in\Sigma^*$ such that $\delta(\{i_0\},v)=\{p\}$, as $v$ can be chosen so that the path labeled by $v$ starting from $i_0$ is simple and does not encounter $p$ twice, and as $p$ is the only state with a non-deterministic transition.
As a consequence, for $u=w_\A(aw_\A)^{d-1}$, we have $\delta(\{i_0\},vau)=\{p_0,p_1,\ldots,p_d\}$, and therefore, $\C_b(\A)$ is in $\D(\A,i_0)$.
Recall that $|\C_b(\A)|=\textrm{lcm}(\ell_0(\A),\ldots,\ell_d(\A))$, we have
\begin{align*}
\proba\left(|\D(\A,i_0)|\geq n^{(d+1)/2}\right)
& \geq 
\proba\left(\ell(\A,p,q)\neq\bot \text{ and }|\D(\A)|\geq n^{(d+1)/2}\right)\\
& \geq 
\proba\left(\ell(\A,p,q)\neq\bot \text{ and }\textrm{lcm}(\ell_0(\A),\ldots,\ell_d(\A))\geq  n^{(d+1)/2}\right)
\end{align*}
By Proposition~\ref{pro:loops}, given $\ell(\A,p,q)\neq\bot$, $(\ell_0(\A),\ldots,\ell_d(\A))$ is a uniform random element of $\intinter{\sqrt{n}}{2\sqrt{n}}^{d+1}$.
Hence Corollary~\ref{cor:pairwise coprime} applies and yields that for some $c>0$ and $n$ sufficiently large,
\[
\proba\left(\textrm{lcm}(\ell_0(\A),\ldots,\ell_d(\A))\geq  n^{(d+1)/2}\mid \ell(\A,p,q)\neq\bot\right) \geq c.
\]
Therefore, we have
\[
\proba\left(|\D(\A)|\geq n^{(d+1)/2}\right)\geq c \ \proba(\ell(\A,p,q)\neq\bot),
\]
Hence by Proposition~\ref{pro:loops}, $|\D(\A)|\geq n^{(d+1)/2}$ with visible probability. Since $p$ is accessible with high probability, we have, as in Lemma~\ref{lm:accessible p}, that both properties hold with visible probability. This concludes the proof.
\end{proof}

\section{State complexity}\label{sec:state complexity}

We are now ready to randomly select which states are final, and to state a result on languages instead of on  transition structures, by studying the state complexity of a language recognized by a random almost deterministic automaton. 

In our model, for every $n$, each state is final with fixed probability $f_n$, which may depend on $n$ as long as it is not too close to either $0$ or $1$: we require that $f_n$ and $1-f_n$ are in $\Omega(\frac1{\sqrt{n}})$. Using a variant of the Birthday Problem again, this ensures that a uniform random set of states of size  $\Theta(\sqrt{n})$ contains both final and non-final states with visible probability.

In the proof of Theorem~\ref{thm:determinisation}, we exhibited the existence with visible probability of $d+1$ occurrences of $b$-cycles in a random almost deterministic transition structure, yielding a large $b$-cycle when applying the powerset construction. We will focus on $b$-cycles in the sequel, as it turns out to be sufficient to prove our main result. We rely on the notion of primitive words, which we now recall.

Let $\Gamma$ be a nonempty finite alphabet. If $w\in\Gamma^\ell$ is a word of length $\ell$, we write $w=w_0\cdots w_{\ell-1}$ and use the convention that all indices are taken modulo $\ell$: for instance $w_\ell$ is the letter $w_0$. A nonempty word $w$ is \emph{primitive} if it is not a non-trivial power of another word: it cannot be written $w=z^k$ for some word $z$ and some $k\geq 2$. If $w$ is primitive, it is easily seen that every circular permutation of $w$ is also primitive. See~\cite{lothaire1997} for a more detailed account on primitive words.

Primitive words appear in our proof with the following observation.
If $\C=(c_0,\ldots,c_{\ell-1})$ is a $b$-cycle of states starting at $c_0$, its \emph{associated binary word} is the size-$\ell$  word $v=v_0\ldots v_{\ell-1}$ of $\{0,1\}^\ell$ where $v_i=1$ if and only if
$c_i$ is a final state. Recall that if we start the same cycle elsewhere, at $c_i$, the associated word $v'=v_i\cdots v_{\ell}v_0\cdots v_{i-1}$ is primitive if and only if $v$ is primitive: reading the associated binary word from any starting state preserves primitivity. A $b$-cycle is said to be \emph{primitive} if one (equivalently, all) of its associated words is (are) primitive. Our study is based on the following statement.
\begin{lem}\label{lm:primitive state complexity}
Let $\A$ be a deterministic automaton on $\Sigma$ and $\alpha\in\Sigma$. If $\C$ is a primitive $\alpha$-cycle of $\A$, then the states of $\C$ are pairwise non-equivalent: the state complexity of the language recognized by $\A$ is at least $|\C|$, the number of states in $\C$.
\end{lem}
\begin{proof}
Let $\A=(n,\delta)$ and 
let $p$ and $q$ be two different states of $\C$. Let $x$ and $y$ be the associated binary words of $\C$ starting at $p$ and $q$, respectively. Let $k$ be the smallest positive integer such that $\delta(p,\alpha^k)=q$ and
let $u$ be the prefix of length $k$ of $x$, and $v$ be the associated suffix: $x=uv$. Then $y=vu$. Assume by contradiction that $p$ and $q$ are equivalent. This implies that $x=y$, as the automata obtained by placing the initial states either on $p$ or $q$ recognize the same elements of $\{\alpha\}^*$. Hence $uv=vu$, and therefore $u$ and $v$ are the power of the same word by a classical result on primitive words~\cite[Prop. 1.3.2 page 8]{lothaire1997}. This is in contradiction with the fact that $\C$ is primitive.
\end{proof}

\subsection{A composition of binary primitive words}

If $w^{(1)}$ and $w^{(2)}$ are two non-empty words of respective lengths $\ell_1$ and $\ell_2$ on the binary alphabet $\{0,1\}$, we denote by $w^{(1)}\uprod w^{(2)}$ the word $w$ of length $\ell = \text{lcm}(\ell_1,\ell_2)$ given by $w_i=1$ if and only if $w^{(1)}_i=1$ or $w^{(2)}_i=1$ (recall that the indices are taken modulo the length of the word). We will see in the sequel that this operation naturally happens when extending the notion of state equivalence from each $b$-cycle to the corresponding $b$-cycle in the powerset construction.

\begin{lem}\label{lm:primitive product}
Let $w^{(1)}$ and $w^{(2)}$ be two primitive words on $\{0,1\}$ of lengths at least $2$ that are coprime. Then the word $w^{(1)}\uprod w^{(2)}$ is primitive.
\end{lem}

\begin{proof}
Let $w=w^{(1)}\uprod w^{(2)}$. Assume by contradiction that there exists some word $z$ and some $k\geq 2$ such that $w=z^k$. {Letting} $p$ be a prime number that divides $k$, we have $w=(z^{k/p})^p$. This
yields that $p$ divides $\ell=\text{lcm}(\ell_1,\ell_2)=\ell_1\times\ell_2$ and that for every non-negative integer $i$, $w_i=w_{i+\ell/p}$ (indices taken modulo
$\ell$). Obviously, $p$ divides either $\ell_1$ or $\ell_2$, but not both as they are coprime. By symmetry, assume that it divides $\ell_1$: $\ell_1=pr$ and
$\ell/p = r\ell_2$. 

Since $w^{(2)}$ has length at least $2$ and is primitive, there exists an index $i_0\in\{0,\ldots,\ell_2-1\}$ such that
$w^{(2)}_{i_0}=0$. Define $i_j=i_0+j\ell_2$, for any $j\geq 0$. 
As indices in $w^{(2)}$ are taken modulo $\ell_2$, we have
$w^{(2)}_{i_j}=0$ for all $j\geq 0$. Therefore, $w^{(1)}_{i_j}=1$
if and only if $w_{i_j} = 1$, so that $w_{i_j}=w^{(1)}_{i_j}$ for all $j\geq 0$. 
As $w_{i_j}=w_{i_j+\ell/p}$ and $\ell/p=r\ell_2$, this yields that $w^{(1)}_{i_j}=w^{(1)}_{i_j+r\ell_2}$ for all $j\geq 0$. Moreover,
$r\ell_2$ is not a multiple of $\ell_1$: let $\alpha\geq 1$ be the largest integer such that $p^\alpha$ divides $\ell_1=pr$, then $p^\alpha$ does not divide $r\ell_2$, since $p$ does not divide $\ell_2$.

Let $s:= r\ell_2 \mod \ell_1$, we just established that $s\neq 0$, so we have
the non-trivial relation $w^{(1)}_{i_j}=w^{(1)}_{i_j+s}$ for all $j\geq 0$. Recall that $i_j=i_0+j\ell_2$. As $\ell_1$ and $\ell_2$ are coprime, the $i_j$'s take all  the values modulo $\ell_1$ when $j$ ranges from $0$ to $\ell_1-1$. Hence, for all $k\in\{0,\ldots,\ell_1-1\}$, $w^{(1)}_{k}= w^{(1)}_{k+s}$, for some {$0<s<\ell_1$}.
This is a contradiction with the fact that $w^{(1)}$ is primitive, concluding the proof.
\end{proof}

\begin{rem}    
Lemma~\ref{lm:primitive product} does not hold if the lengths are not coprime. For instance, if $w^{(1)}=011111$ and $w^{(2)}=1011$,
then $w^{(1)}\uprod w^{(2)}=\underbrace{1\ldots 1}_{12\text{ times}}$, which is not primitive.
\end{rem}

\subsection{Primitivity of random words} From a probabilistic point of view, it is well known~\cite{lothaire1997} that a uniform random word is primitive with exponentially high probability. We rely on the following finer result.
\begin{lemC}[
\cite{sven16}]\label{lm:primitive proba old}
Let $\mu$ be a probability measure on $\{0,1\}^n$ such that
$\mu(0^n) = \mu(1^n) = 0$ and such that two words with the same number of $0$'s have  same probability. Then the probability that a word is not primitive under $\mu$ is at most $\frac2n$.  
\end{lemC}
We adapt it to our needs as follows:
\begin{cor}\label{cor:proba primitive}
Let $f_n$ be a sequence of real numbers in $(0,1)$ such that
$f_n = \Omega(\frac1{\sqrt{n}})$ and $1-f_n = \Omega(\frac1{\sqrt{n}})$. {Consider the distribution on random words of length $n$ where each letter is chosen independently to be $1$ with probability $f_n$ and $0$ with probability $1-f_n$. 
There exists  $c>0$ such that, for $n$ sufficiently large and for any integer $\ell\geq\sqrt{n}$, a random  binary word of length $\ell$  is primitive with probability at least $c$. }
\end{cor}

\begin{proof}
Let $X$ be the event that $w=0^\ell$ or $w=1^\ell$. We have $\mathbb{P}(X) = f_n^\ell+(1-f_n)^\ell$. Since changing the $0$'s in $1$'s and the $1$'s in $0$'s preserves primitivity, we can assume by symmetry that $f_n\leq \frac12$.
By hypothesis, there exists some constant
$\beta>0$ such that $\frac{\beta}{\sqrt{n}}\leq f_n$ and 
$\frac{\beta}{\sqrt{n}}\leq 1-f_n$ hence, as $f_n\leq \frac12$, we have 
\[
f_n^\ell \leq \frac1{2^{\sqrt{n}}}\text{ and }(1-f_n)^\ell \leq \left(1-\frac{\beta}{\sqrt{n}}\right)^{\sqrt{n}}.
\]
Since $(1-\frac{\beta}{\sqrt{n}})^{\sqrt{n}} = e^{-\beta}+O(\frac1{\sqrt{n}})$, there exists some constant $\delta<1$ such that $\mathbb{P}(X) \leq \delta$, for $n$ sufficiently large. 

Let $W$ be a random word under our distribution. For any $w\in\{0,1\}^\ell$, the conditional probability that $W$ values $w$ given that $W\notin\{0^\ell,1^\ell\}$ is
\[
\mathbb{P}(W=w\mid \overline{X}) = 
    \begin{cases}
    0 & \text{if }w=0^\ell \text{ or } w=1^\ell,\\
    \frac{\mathbb{P}(W=w)}{1-\mathbb{P}(X)}&\text{otherwise}.
    \end{cases}
\]
Hence we are in the settings of Lemma~\ref{lm:primitive proba old}, and the probability that $w$ is not primitive, given that $w\notin\{0^\ell,1^\ell\}$ is at most $\frac{2}\ell$. Moreover, we have:
\[
\mathbb{P}(w \text{ not primitive}) = \mathbb{P}(X) +
\mathbb{P}(w \text{ not primitive}\mid\overline{X})\mathbb{P}(\overline{X})
\leq \delta  + \frac{2}\ell,
\]
where $\delta  + \frac{2}\ell<1$, for $n$ large enough. This concludes the proof.
\end{proof}

\subsection{Finalizing the proof of Theorem~\ref{th:main}}

If $\A=(n,\delta)$ is a transition structure such that $\ell(\A,p,q)\neq\bot$, recall that $\C_b(\A)$ is the $b$-cycle of its powerset construction $\D(\A)$ starting from its set of starting states $P=\{p_0,\ldots,p_d\}$. Let $j_0$ be a positive integer such that $\delta(p_i,b^{j_0})$ is in the $b$-cycle of its $b$-thread for all $i\in\{0,\ldots,d\}$, and let $P_0=(\delta(p_0,b^{j_0}),\ldots,\delta(p_d,b^{j_0}))$. 
We use a tuple for $P_0$ as it is easier for the upcoming arguments, but since the $b$-threads of $\A$ are pairwise disjoint, and since we only read $b$'s from $P_0$ in the sequel, they can be identified with the associated sets of $\D(\A)$. 
Let us write $c_j^{(i)}=\delta(p_i,b^{j_0+j})$ for all $i\in\{0,\ldots,d\}$ and all $j\geq 0$. The $b$-cycle $\C_i$ of the $i$-th $b$-thread has length $\ell_i$, so if $j=j'\mod\ell_i$, then $c_j^{(i)}=c_{j'}^{(i)}$. 
Moreover, the $b$-cycle $\C_b(\A)$ is made of the $(c_j^{(0)},\ldots,c_j^{(d)})$ for $j\geq 0$.

We now consider a set of final state $F\subseteq[n]$. For this $F$, and for every 
$i\in\{0,\ldots,d\}$, let $w^{(i)}$ denote the binary word associated to $\C_i$, starting at $c_0^{(i)}$. Let $w$ denote the binary word associated with $\C_b(\A)$, starting at $P_0$. Since a state $X$ in the powerset construction is final if and only if one of its state is in $F$, by a direct induction we have
\[
w = w^{(0)}\odot w^{(1)}\odot\ldots\odot w^{(d)}.
\] 
This is depicted in Figure~\ref{fig:primitive} for two $b$-threads.

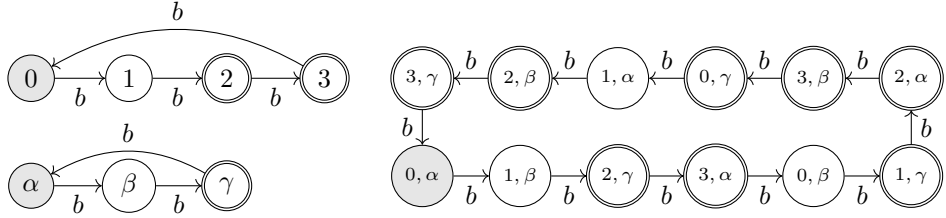
\begin{figure}[h]
\centering
    {\small
\begin{tikzpicture}[scale=.65]
	\node[draw,circle,fill=black!10] (s0) at (0,0) {$0$};
	\node[draw,circle] (s1) at (2,0) {$1$};
	\node[draw,circle,accepting] (s2) at (4,0) {$2$};
	\node[draw,circle,accepting] (s3) at (6,0) {$3$};
	
	\draw[->] (s0) -- node[below]{$b$} (s1);
	\draw[->] (s1) -- node[below]{$b$} (s2);
	\draw[->] (s2) -- node[below]{$b$} (s3);
	\draw[->] (s3) edge[bend right=30] node[above]{$b$} (s0);
	
	\node[draw,circle,fill=black!10] (t0) at (0,-2.2) {$\alpha$};
	\node[draw,circle] (t1) at (2,-2.2) {$\beta$};
	\node[draw,circle,accepting] (t2) at (4,-2.2) {$\gamma$};
	
	\draw[->] (t0) -- node[below]{$b$} (t1);
	\draw[->] (t1) -- node[below]{$b$} (t2);
	\draw[->] (t2) edge[bend right=30] node[above]{$b$} (t0);
	
	\node[draw,circle,fill=black!10] (p00) at (8,-2) {\tiny $0,\alpha$};
	\node[draw,circle] (p11) at (10,-2) {\tiny$1,\beta$};
	\node[draw,circle,accepting] (p22) at (12,-2) {\tiny$2,\gamma$};
	\node[draw,circle,accepting] (p30) at (14,-2) {\tiny$3,\alpha$};
	\node[draw,circle] (p01) at (16,-2) {\tiny$0,\beta$};
	\node[draw,circle,accepting] (p12) at (18,-2) {\tiny$1,\gamma$};
	
	\node[draw,circle,accepting] (p32) at (8,0) {\tiny $3,\gamma$};
	\node[draw,circle,accepting] (p21) at (10,0) {\tiny$2,\beta$};
	\node[draw,circle] (p10) at (12,0) {\tiny$1,\alpha$};
	\node[draw,circle,accepting] (p02) at (14,0) {\tiny$0,\gamma$};
	\node[draw,circle,accepting] (p31) at (16,0) {\tiny$3,\beta$};
	\node[draw,circle,accepting] (p20) at (18,0) {\tiny$2,\alpha$};

	\draw[->] (p00) -- node[below]{$b$} (p11);
	\draw[->] (p11) -- node[below]{$b$} (p22);
	\draw[->] (p22) -- node[below]{$b$} (p30);
	\draw[->] (p30) -- node[below]{$b$} (p01);
	\draw[->] (p01) -- node[below]{$b$} (p12);
	\draw[->] (p12) -- node[right]{$b$} (p20);
	
	\draw[->] (p20) -- node[above]{$b$} (p31);
	\draw[->] (p31) -- node[above]{$b$} (p02);
	\draw[->] (p02) -- node[above]{$b$} (p10);
	\draw[->] (p10) -- node[above]{$b$} (p21);
	\draw[->] (p21) -- node[above]{$b$} (p32);
	\draw[->] (p32) -- node[left]{$b$} (p00);
\end{tikzpicture}
}
    \caption{On the left, two primitive $b$-cycles (accepting states are denoted by double circles) whose associated words are $0011$ (top) and $001$ (bottom), starting at $0$ and $\alpha$, respectively. On the right, the $b$-cycle of $\{0,\alpha\}$ of associated word $0011\uprod 001 = 001101111011$, which is primitive by Lemma~\ref{lm:primitive product}.\label{fig:primitive}}
\end{figure}

{
To finish the proof of Theorem~\ref{th:main}, we follow the proof of Theorem~\ref{thm:determinisation} and add the independent choice of which states are final. By Corollary~\ref{cor:proba primitive}, all the $b$-cycles used to build $\C_b(\A)$ are also primitive, still with visible probability. Hence 
by Lemma~\ref{lm:primitive product}, the product cycle $\C_b(\A)$ is also primitive with visible probability. This concludes the proof by Lemma~\ref{lm:primitive state complexity}.
}

\section{{Remark on dense random DFAs}}\label{sec:dense models}

{
In this short section, we  illustrate the degenerate nature of \emph{dense} random NFAs by considering the model with a unique initial state and in which each transition is added independently with a fixed probability $p \in (0,1)$. 
A uniform dense random NFA (with one initial state) corresponds to $p=\frac{1}{2}$. 
The following proof is very similar to the analysis of the diameter of random dense undirected graphs~\cite[Corollary 10.11]{Bollobas01}.
}

{
Consider a random $n$-state NFA $\A$ under this model, with initial state $i_0$ and transition function $\delta$. Let $w=w_0w_1$ be any word of length $2$ on $\Sigma$ and let $q$ be any state of $\A$. We bound from above the probability that $q\notin\delta(i_0,w)$ as follows~: 
\begin{align*}
\proba\left(q\notin\delta(i_0,w)\right) & =
\proba\left(\forall r\in [n],\ r\notin\delta(i_0,w_0)
\text{ or } q\notin\delta(r,w_1)\right) \\
&\leq \proba\left(\forall r\in [n]\setminus\{i_0\},\ r\notin\delta(i_0,w_0)
\text{ or } q\notin\delta(r,w_1)\right) \\
& = \prod_{r\in [n]\setminus\{i_0\}} \proba\left(r\notin\delta(i_0,w_0)
\text{ or } q\notin\delta(r,w_1)\right) \\
& = \prod_{r\in [n]\setminus\{i_0\}} \left(1- \proba\left(r\in\delta(i_0,w_0)
\text{ and } q\in\delta(r,w_1)\right)\right)\\
& = \prod_{r\in [n]\setminus\{i_0\}} (1-p^2) = (1-p^2)^{n-1}.
\end{align*}
Thus, by the union bound we  obtain that with high probability, the image of $\{i_0\}$ by any word of length $2$ is the set of all states of $\A$. We therefore obtain the following statement, which establishes the degeneracy of the dense model.
\begin{prop}\label{pro:dense}
The accessible subset automaton of a random $n$-state dense NFA with one initial state has at most $|\Sigma|+2$ states with high probability. Hence, with high probability, the state complexity of the language is bounded by a constant independent of $n$ and $p$.
\end{prop}
}

\section{Conclusion and discussion}

Our main theorem states that the state complexity of a random almost deterministic automaton is greater than $n^d$ with probability at least $c_d>0$ for $n$ sufficiently large. One can wonder how small the constant $c_d$ is and for which sizes the lower-bound holds. As we said in the introduction, we did not try to estimate $c_d$ nor did we try to optimize its value in this article.  Since the powerset construction quickly generates very large automata which would need to be minimized, a proper experimental study does not seem feasible. However, we did generate $1000$ almost deterministic transition structures with $n=100$ states and apply the accessible powerset construction: in $78.6\%$ of the $1000$ cases  the output had more than $n^3$ states. This would lead us to guess that even if the constant $c_3$ that can be derived from our proof is very small, 
combinatorial explosion does occur frequently in practice. 

Also, as noticed above, in our settings it is certain that the property does not hold with high probability, as there is an asymptotically constant probability that the source of the added transition is not accessible. However, this probability is roughly $20.4\%$, not too far from what we obtained in our experiment on size-$100$ structures: it is very possible that if we condition the source of the added transition to be accessible, then our result holds with high probability. However, our proof techniques, based on an intensive use of the Birthday Problem  cannot prove this: completely new ideas are necessary to establish such a result.

{It is natural to ask whether our result could be strengthen from super-polynomial size to exponential size. We do not know if this generalization holds, but it seems out of reach of the techniques developed in this paper.}

Another natural direction is to consider the case when there are \emph{few} final states, as $\Theta(\sqrt{n})$ final states may be considered too large for a random deterministic automaton. The extreme case is to allow exactly one final state by choosing it uniformly at random. If we do so, our analysis using primitive words fails: with high probability the $b$-cycles we built have no final state at all, and neither has the associated $b$-cycle $\C$ in the powerset construction. However, we are confident that our techniques can be used to capture this distribution: by studying the paths ending in this final state, we should be able to find for each $b$-cycle $\C_i$ a word $w_i$ that maps exactly one state to the final state, and such that the $w_i$'s are all different. This would be enough to establish that the states of $\C$ are pairwise non-equivalent and prove the conjecture. Completely formalizing and proving this idea  is an ongoing work.

\section*{Acknowledgment}
  \noindent The authors would like to thank the anonymous reviewers for their valuable feedback. This work was partially supported by the French ANR grant ASPAG (ANR-17-CE40-0017). 

\bibliographystyle{alphaurl}
\bibliography{biblio}
\end{document}